\documentclass[11pt]{article}

\usepackage[utf8]{inputenc}
\usepackage[T1]{fontenc}
\usepackage{amsthm,amsmath,amsfonts,amssymb}
\usepackage[margin=1.5in]{geometry}
\usepackage{authblk}
\usepackage{algorithm}
\usepackage[noend]{algpseudocode}
\usepackage{todonotes}
\usepackage{subcaption}
\usepackage{float}
\usepackage[skip=2pt]{caption}
\usepackage{enumitem}
\usepackage{graphicx}
\usepackage{hhline}
\usepackage{diagbox}
\usepackage{booktabs}
\usepackage{subcaption}
\usepackage{rotating}
\usepackage{thmtools}
\usepackage{thm-restate}

\newcommand{\cT}{\mathcal{T}}
\newcommand{\N}{\mathcal{N}}
\newcommand{\cO}{\mathcal{O}}
\newcommand{\R}{\mathbb{R}}
\newcommand{\rand}{\textsf{Rand}}
\newcommand{\lowpair}{\textsf{LowPair}}
\newcommand{\trivialrand}{\textsf{TrivialRand}}
\newcommand{\ML}{\textsf{ML}}
\newcommand{\trivML}{\textsf{TrivialML}}
\newcommand{\featimp}{\textsf{FeatImp}}
\newcommand{\OPT}{\textsf{OPT}}

\newcommand{\cherries}{\textsf{cherryfeatures}}
\newcommand{\Treechild}{\textsf{TreeChild}}
\newcommand{\Hybro}{\textsf{Hybroscale}}

\newcommand{\cherrylist}{\textsf{cherrylist}}

\newtheorem{definition}{Definition}
\newtheorem{theorem}{Theorem}
\newtheorem{lemma}{Lemma}
\newtheorem{example}{Example}
\newtheorem{corollary}{Corollary}

\title{Constructing Phylogenetic Networks via Cherry Picking and Machine Learning\thanks{The work in this paper is supported by: the Netherlands Organisation for Scientific Research (NWO) through project OCENW.GROOT.2019.015 ``Optimization for and with Machine Learning (OPTIMAL)'' and Gravitation-grant NETWORKS-024.002.003; the PANGAIA and ALPACA projects that have received funding from the European Union’s Horizon 2020 research and innovation programme under the Marie Skłodowska-Curie grant agreements No 872539 and 956229, respectively; and the MUR - FSE REACT EU - PON R\&I 2014-2020.}}

\author[1]{Giulia Bernardini}
\author[2]{Leo van Iersel}
\author[2]{Esther Julien}
\author[3,4,5]{Leen Stougie}

\affil[1]{University of Trieste, Trieste, Italy}
\affil[2]{Delft Institute of Applied Mathematics, Delft, The Netherlands}
\affil[3]{CWI, Amsterdam, The Netherlands}
\affil[4]{Vrije Universiteit, Amsterdam, The Netherlands}
\affil[5]{INRIA-Erable, France}
\date{}

\begin{document}
\maketitle

\begin{abstract} 
Combining a set of phylogenetic trees into a single phylogenetic network that explains all of them is a fundamental challenge in evolutionary studies. Existing methods are computationally expensive and can either handle only small numbers of phylogenetic trees or are limited to severely restricted classes of networks.
In this paper, we apply the recently-introduced theoretical framework of cherry picking to design a class of efficient heuristics that are guaranteed to produce a network containing each of the input trees, for datasets consisting of binary trees. Some of the heuristics in this framework are based on the design and training of a machine learning model that captures essential information on the structure of the input trees and guides the algorithms towards better solutions.
We also propose simple and fast randomised heuristics that prove to be very effective when run multiple times.

Unlike the existing exact methods, our heuristics are applicable to datasets of practical size, and the experimental study we conducted on both simulated and real data shows that these solutions are qualitatively good, always within some small constant factor from the optimum.
Moreover, our machine-learned heuristics are one of the first applications of machine learning to phylogenetics and show its promise.
\end{abstract}

\section{Introduction}
\label{sec:intro}
Phylogenetic networks describe the evolutionary relationships between different objects: for example, genes, genomes, or species. 
One of the first and most natural approaches to constructing phylogenetic networks is to build a network from a set of gene trees. In the absence of incomplete lineage sorting, the constructed network is naturally required to ``display'', or embed, each of the gene trees. In addition, following the parsimony principle, a network assuming a minimum number of reticulate evolutionary events (like hybridization or lateral gene transfer) is often sought. Unfortunately, the associated computational problem, called \textsc{Hybridization}, is NP-hard even for two 
binary input trees~\cite{bordewich2007computing}, and indeed existing solution methods do not scale well with problem size.

For a long time, research on this topic was mostly restricted to inputs consisting of two trees. Proposed algorithms for multiple trees were either completely impractical or ran in reasonable time only for very small numbers of input trees. This situation changed drastically with the introduction of so-called cherry-picking sequences~\cite{linzcherrypicking}.
This theoretical setup opened the door to solving instances consisting of many input trees like most practical datasets have. Indeed, a recent paper showed that this technique can be used to solve instances with up to~$100$ input trees to optimality~\cite{van2019practical}, although it was restricted to binary trees all having the same leaf set and to so-called ``tree-child'' networks. Moreover, its running time has a (strong) exponential dependence on the number of reticulate events.

In this paper, we show significant progress towards a fully practical method by developing 
a heuristic framework based on cherry picking comprising very fast randomised heuristics and other slower but more accurate heuristics guided by machine learning.
Admittedly, our methods are not yet widely applicable since they are still restricted to binary trees. 
However, our set-up is made in such a way that it  may be extendable to general trees.  

Despite their limitations, we see our current methods already as a breakthrough as they  are not restricted to tree-child networks and scale well with the number of trees, the number of taxa and the number of reticulations. In fact, we experimentally show that our heuristics can easily handle sets of 100 trees in a reasonable time: the slowest machine-learned method takes 4 minutes on average for sets consisting of 100 trees with 100 leaves each, while the faster, randomised heuristics already find feasible solutions in 2 seconds for the same instances. As the running time of the fastest heuristic depends at most quadratically on the number of input trees, linearly on the number of taxa, and linearly on the output number of reticulations, we expect it to be able to solve much larger instances still in a reasonable amount of time.

In addition, in contrast with the existing algorithms, our methods can be applied to trees with different leaf sets, although they have not been specifically optimized for this kind of input. Indeed, we experimentally assessed that our methods give qualitatively good results only when the leaf sets of the input trees have small differences in percentage (up to 5-15\%); when the differences are larger, they return feasible solutions that are far from the optimum.

Some of the heuristics we present are among the first applications of machine learning in phylogenetics and show its promise. In particular, we show that crucial features of the networks generated in our simulation study can be identified with very high test accuracy ($99.8\%$) purely based on the trees displayed by the networks. 

It is important to note at this point that no method is able to reconstruct any specific network from displayed trees as networks are, in general, not uniquely determined by the trees they display~\cite{pardi}.  In addition, in some applications, a phenomenon called ``incomplete lineage sorting'' can cause gene trees that are not displayed by the species network~\cite{yu2011coalescent}, and hence our methods, and other methods based on the Hybridization problem, are not (directly) applicable to such data.

We focus on \emph{orchard} networks (also called \emph{cherry picking} networks), which are precisely those networks that can be drawn as a tree with additional horizontal arcs~\cite{van2021orchard}. Such horizontal arcs can for example correspond to lateral gene transfer (LGT), hybridization and recombination events. Orchard networks are broadly applicable:
in particular, the orchard network class is much bigger than the class of tree-child networks, to which the most efficient existing methods are limited~\cite{hybroscale}.

\paragraph*{Related work.}
Previous practical algorithms for \textsc{Hybridization} include PIRN~\cite{PIRN}, PIRNs~\cite{PIRNs} and \Hybro{}~\cite{hybroscale}, exact methods that are only applicable to (very) small numbers of trees and/or to trees that can be combined into a network with a (very) small reticulation number. Other methods such as \textsc{PhyloNet}~\cite{phylonet} and \textsc{PhyloNetworks}~\cite{solis2017phylonetworks} also construct networks from trees but have different premises and use completely different models.

The theoretical framework of cherry picking was introduced in~\cite{temporalcherrypicking} (for the restricted class of temporal networks) and~\cite{linzcherrypicking} (for the class of tree-child networks) and was later turned into algorithms for reconstructing tree-child~\cite{van2019practical} and temporal~\cite{borst2022new} networks. These methods can handle instances containing many trees but do not scale well with the number of reticulations, due to an exponential dependence. 
The class of orchard networks, which is based on cherry picking, was introduced in~\cite{semple2021trinets} and independently (as cherry-picking networks) in~\cite{janssen2021cherry}, although their practical relevance as trees with added horizontal edges was only discovered later~\cite{van2021orchard}.

The applicability of machine-learning techniques to phylogenetic problems has not yet been fully explored, and to the best of our knowledge existing work is mainly limited to phylogenetic tree inference~\cite{azouri2021harnessing,DBLP:conf/icmlc2/ZhuC21} and to testing evolutionary hypotheses~\cite{kumar2021evolutionary}.

\subparagraph{Our contributions.}
We introduce  \textsc{Cherry Picking Heuristics (CPH)}, a class of heuristics to combine a set of binary phylogenetic trees into a single binary phylogenetic network based on cherry picking. 
We define and analyse several heuristics in the CPH class, all of which are guaranteed to produce feasible solutions to \textsc{Hybridization} and all of which can handle instances of practical size (we run experiments on tree sets of up to 100 trees with up to 100 leaves which were processed in on average 4 minutes by our slowest heuristic).

Two of the methods we propose are simple but effective randomised heuristics that proved to be extremely fast and to produce good solutions when run multiple times.
The main contribution of this paper consists in a
machine-learning model that potentially captures essential information about the structure of the input set of trees. 
We trained the model on different extensive sets of synthetically generated data and applied it to guide our algorithms towards better solutions.
Experimentally, we show that the two machine-learned heuristics we design yield good results when applied to both synthetically generated and real data.

We also analyse our machine-learning model to identify the most relevant features and design a non-learned heuristic that is guided by those features only. Our experiments show that this heuristic leads to reasonably good results without the need to train a model. This result is interesting per se as it is an example of how machine learning can be used to guide the design of classical algorithms,  which are not biased towards certain training data. 

A preliminary version of this work appeared in~\cite{DBLP:conf/wabi/BernardiniIJS22}. Compared to the preliminary version, we have added the following material: (i), we defined a new non-learned heuristic based on important features and experimentally tested it (Section~\ref{sec:new_heuristic}); (ii), we extended the experimental study to data generated from non-orchard networks (Section~\ref{exp:ZODS}), data generated from a class of networks for which the optimum number of reticulations is known (Section~\ref{exp:normal}) and to input trees with different leaf sets (Section~\ref{exp:missing_leaves}); and (iii), we provided a formal analysis of the time complexity of all our methods (Section~\ref{sub:time_complexity}) and conducted experiments on their scalability (Section~\ref{exp:scalability}).

\section{Preliminaries}
\label{sec:preliminaries}

A \emph{phylogenetic network} $N=(V,E,X)$ on a set of taxa $X$ is a directed acyclic graph $(V,E)$ with a single \emph{root} with in-degree 0 and out-degree 1, and the other nodes with either (i) in-degree 1 and out-degree $k>1$ (\emph{tree nodes}); (ii) in-degree $k>1$ and out-degree 1  (\emph{reticulations}); or (iii) in-degree 1 and out-degree 0 (\emph{leaves}).
The leaves of $N$ are biunivocally labelled by $X$. A surjective map $\ell:E\rightarrow \mathbb{R}^{\ge 0}$ may assign a nonnegative \emph{branch length} to each edge of $N$.
We will denote by $[1,n]$ the set of integers $\{1,2,...,n\}$.
Throughout this paper, we will only consider binary networks (with $k=2$), and we will identify the leaves with their labels.
We will also often drop the term ``phylogenetic'', as all the networks considered in this paper are phylogenetic networks.
The \emph{reticulation number} $r(N)$ of a network $N$ is 
\(\sum_{v\in V}\max\left(0,d^{-}(v)-1\right),\)
where $d^-(v)$ is the in-degree of $v$. A network $T$ with $r(T)=0$ is a \emph{phylogenetic tree}.
It is easy to verify that binary networks with $r(N)$ reticulations have $|X|+r(N)-1$ tree nodes.

\subparagraph{Cherry-picking.}
We denote by $\N$ a set of networks and by $\cT$ a set of trees. An {\it ordered} pair of leaves $(x,y),~x\neq y$, is a \emph{cherry} in a network if $x$ and $y$ have the same parent; $(x,y)$ is a \emph{reticulated cherry} if the parent $p(x)$ of $x$ is a reticulation, and $p(y)$ is a tree node and a parent of $p(x)$ (see Figure~\ref{fig:cherry_picking_example}). A pair is \emph{reducible} if it is either a cherry or a reticulated cherry.
Notice that trees have cherries but no reticulated cherries.

\emph{Reducing} (or \emph{picking}) a cherry $(x,y)$ in a network $N$ (or in a tree) is the action of deleting $x$ and replacing the two edges $(p(p(x)),p(x))$ and $(p(x),y)$ with a single edge $(p(p(x)),y)$ (see Figure \ref{fig:cherry_picking}). If $N$ has branch lengths, the length of the new edge is $\ell(p(p(x)),y)=\ell(p(p(x)),p(x))+\ell(p(x),y)$. A reticulated cherry $(x,y)$ is reduced (picked) by deleting the edge $(p(y),p(x))$ and replacing the other edge $(z,p(x))$ incoming to $p(x)$, and the consecutive edge $(p(x),x)$, with a single edge $(z,x)$. The length of the new edge is $\ell(z,x)=\ell(z,p(x))+\ell(p(x),x)$ (if $N$ has branch lengths). Reducing a non-reducible pair has no effect on $N$. In all cases, the resulting network is denoted by~$N_{(x,y)}$: we
say that $(x, y)$ affects $N$ if $N\neq N_{(x, y)}$. 

\begin{figure}[h]
    \centering
        \subfloat[]{
        \includegraphics[height=2.7cm]{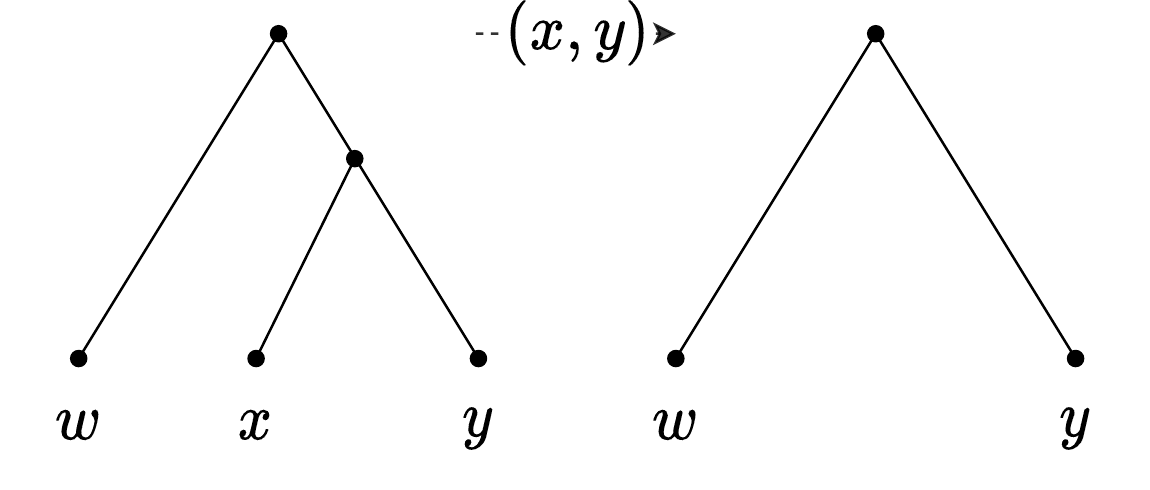} 
       \label{fig:cherry_picking}
       }
       \subfloat[]{
       \includegraphics[height=2.7cm]{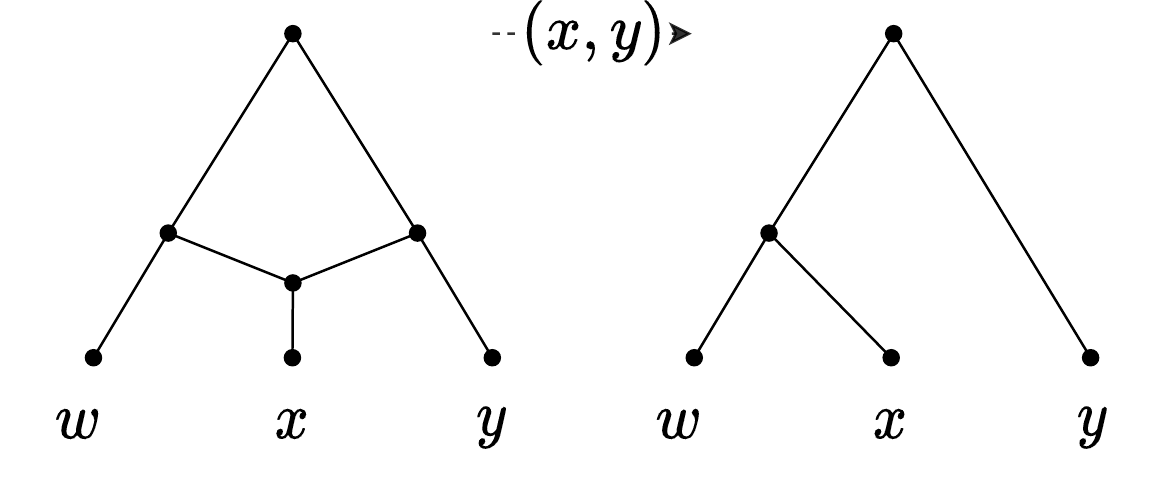} 
       \label{fig:retic_picking}}
    \caption{$(x, y)$ is picked in two different networks. In \textbf{(a)} $(x, y)$ is a cherry, and in \textbf{(b)} $(x, y)$ is a reticulated cherry. After picking, degree-two nodes are replaced by a single edge.} 
    \label{fig:cherry_picking_example}
\end{figure}

Any sequence $S=(x_1,y_1),\ldots,(x_n,y_n)$ of ordered leaf pairs, with $x_i\neq y_i$ for all $i$, is a \emph{partial cherry-picking sequence}; $S$ is a cherry-picking sequence (CPS) if, for each $i<n$, $y_i\in\{x_{i+1},\ldots,x_n,y_n\}$. 
Given a network~$N$ and a (partial) CPS~$S$, we denote by~$N_S$ the network obtained by reducing in $N$ each element of~$S$, in order. We denote $S\circ (x,y)$ the sequence obtained by appending pair $(x,y)$ at the end of $S$.
We say that~$S$ fully reduces $N$ if~$N_S$ consists of the root with a single leaf. $N$ is an \emph{orchard network} (ON) if there exists a CPS that fully reduces it, and it is \emph{tree-child} if every non-leaf node has at least one child that is a tree node or a leaf. A \emph{normal} network is a tree-child network such that, in addition, the two parents of a reticulation are always incomparable, i.e., one is not a descendant of the other. If $S$ fully reduces all $N\in\N$, we say that $S$ fully reduces $\N$. In particular, in this paper we will be interested in CPS which fully reduce a set of trees $\cT$ consisting of $|\cT|$ trees of total size $||\cT||$.
\subparagraph{Hybridization.}
The Hybridization problem can be thought of as the computational problem of combining a set of phylogenetic trees into a network with the smallest possible reticulation number, that is, to find a network that displays each of the input trees in the sense specified by Definition~\ref{def:subnetwork}, below. See Figure \ref{fig:hybridization} for an example. The definition describes not only what it means to display a tree but also to display another network, which will be useful later. 

\begin{figure}[h]
    \centering
       \subfloat[
       ]{ \includegraphics[height=2.7cm]{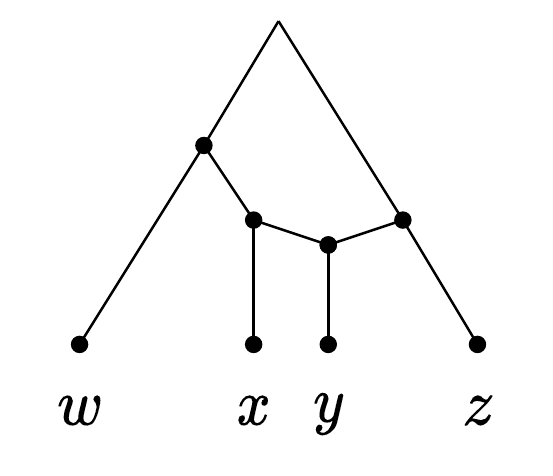} 
       \label{fig:network}}
        \subfloat[
        ]{
        \includegraphics[height=2.7cm]{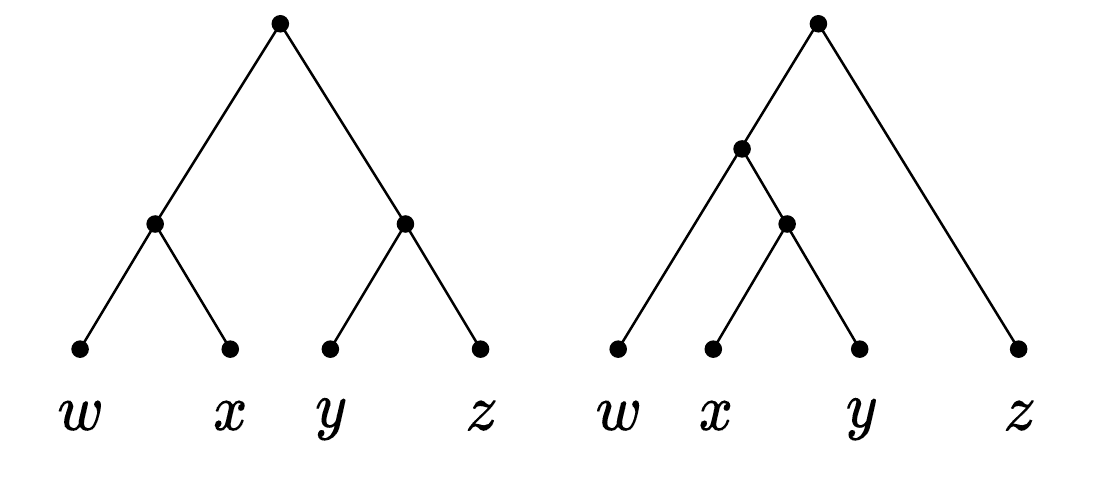} 
        \label{fig:forest}}
    \caption{The two trees in \textbf{(b)} are displayed in the network \textbf{(a)}.} 
    \label{fig:hybridization}
\end{figure}

\begin{definition}\label{def:subnetwork}
Let~$N=(V,E,X)$ and $N'=(V',E',X')$ be networks on the sets of taxa $X$ and ~$X'\subseteq X$, respectively.
The network $N'$ is \emph{displayed} in $N$ if there is an \emph{embedding} of $N'$ in $N$: an injective map of the nodes of~$N'$ to 
the nodes of~$N$, and of the edges of~$N'$ to edge-disjoint paths of~$N$, such that the mapping of the edges respects the mapping of the nodes, and the mapping of the nodes respects the labelling of the leaves.
\end{definition}
We call \emph{exhaustive} a tree displayed in $N=(V,E,X)$ with the whole $X$ as a leaf set.
Note that Definition~\ref{def:subnetwork} only involves the topologies of the networks, disregarding possible branch lengths. In the following problem definition, the input trees may or may not have branch lengths, and the output is a network without branch lengths. We allow branch lengths for the input because they will be useful for the machine-learned heuristics of Section~\ref{sec:ML}. 

\medskip
\fbox{
\parbox{.93\linewidth}{
{\sc Hybridization}\\
{\bf Input:} A set of phylogenetic trees $\mathcal{T}$ on a set of taxa $X$.\\
{\bf Output:} A network displaying $\mathcal{T}$ with minimum possible reticulation number.}
}

\section{Solving the Hybridization Problem via Cherry-Picking Sequences}\label{sec:Hybrid_via_cps}

We will develop heuristics for the Hybridization problem using cherry-picking sequences that fully reduce the input trees, leveraging the following result by Janssen and Murakami.

\begin{theorem}[\cite{janssen2021cherry}, Theorem~3]\label{the:BNBRedimpliesCon}
Let $N$ be a binary orchard network, and $N'$ a (not necessarily binary) orchard network on sets of taxa $X$ and $X'\subseteq X$, respectively. 
If a minimum-length CPS~$S$ that fully reduces $N$ also fully reduces $N'$, then $N'$ is displayed in $N$. 
\end{theorem}
Notice that {\sc Hybridization} remains NP-hard for binary orchard networks.
For binary networks we have the following lemma, a special case of \cite[Lemma 1]{janssen2021cherry}.
\begin{lemma}\label{lem:binary_reconstr}
Let $N$ be a binary network, and let $(x,y)$ be a reducible pair of $N$. Then reducing $(x,y)$ and then adding it back to $N_{(x,y)}$ results in $N$. 
\end{lemma}

Note that Lemma~\ref{lem:binary_reconstr} only holds for binary networks: in fact, there are different ways to add a pair to a non-binary network, thus the lemma does not hold unless a specific rule for adding pairs is specified (inspect \cite{janssen2021cherry} for details).
Theorem~\ref{the:BNBRedimpliesCon} and Lemma~\ref{lem:binary_reconstr} provide the following approach
for finding a feasible solution to \textsc{Hybridization}: 
find a CPS $S$ that fully reduces all the input trees, and then uniquely reconstruct the binary orchard network $N$ for which $S$ is a minimum-length CPS, by processing $S$ in the reverse order.
$N$ can be reconstructed from $S$ using one of the methods underlying Lemma~\ref{lem:binary_reconstr} proposed in the literature, e.g., in~\cite{janssen2021cherry} (illustrated in Figure \ref{fig:network_cps}) or in~\cite{van2019practical}.  The following lemma relates the length of a CPS $S$ and the number of reticulations of the network constructed from $S$.

\begin{lemma}[\cite{van2021unifying}]\label{lem:CPS-retic}
 Let $S$ be a CPS on a set of taxa $X$. The number of reticulations of the network $N$ reconstructed from $S$ is $r(N) = |S| - |X| + 1$.    
\end{lemma}

\begin{figure}[t]
    \centering
   \includegraphics[height=2.7cm]{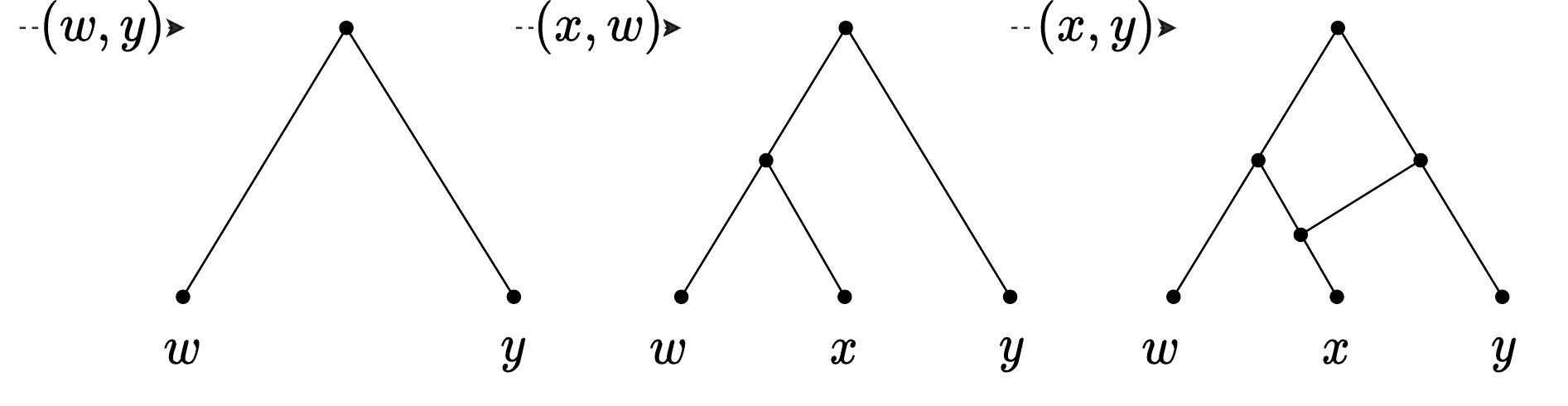}
    \caption{The ON reconstructed from the sequence $S = (x, y), (x, w), (w, y)$. The pairs are added to the network in reverse order: if the first element of a pair is not yet in the network, it is added as a cherry with the second element (see the pair~$(x,w)$). Otherwise, a reticulation is added above the first element with an incoming edge from a new parent of the second element (see the pair~$(x,y)$).} 
    \label{fig:network_cps}
\end{figure}
In the next section we focus on the first part of the heuristic: producing a CPS that fully reduces a given set of phylogenetic trees. 

\subsection{Randomised Heuristics}\label{sub:random_heu}
We define a class of randomised heuristics that construct a CPS by picking one reducible pair of the input set $\cT$ at a time and by appending this pair to a growing partial sequence, as described in Algorithm~\ref{alg:CPH} (the two subroutines \textsf{PickNext} and \textsf{CompleteSeq} will be later described in details). We call this class CPH (for Cherry-Picking Heuristics).
Recall that $\cT_S$ denotes the set of trees $\cT$ after reducing all trees with a (partial) CPS $S$.

\begin{algorithm}[h]
\small
   \caption{\textsc{CPH}}  \label{alg:CPH}
\begin{algorithmic}[1]
\Statex \textbf{INPUT:} A set $\cT$ of phylogenetic trees
 \Statex \textbf{OUTPUT:} A CPS reducing $\cT$.
    \State $S\gets \emptyset$;
    \While{there is a reducible pair in $\cT_S$}\label{line:startwhile}
        \State $(x,y)\gets\textsf{PickNext}(\cT_S)$;\label{line:picknext}
        \State $S\gets S\circ(x,y)$;
        \State Reduce $(x,y)$ in all trees of $\cT_S$;\label{line:endwhile}
    \EndWhile
    \State $S\gets\textsf{CompleteSeq}(S)$;\label{line:completeCPS}
    \State \Return $S$;
  \end{algorithmic}
\end{algorithm}
The while loop at lines~\ref{line:startwhile}-\ref{line:endwhile} produces, in general, a partial CPS $S$, as shown in Example~\ref{ex:partialCPS}. To make it into a CPS, the subroutine \textsf{CompleteSeq} at line~\ref{line:completeCPS} appends at the end of  $S$ a sequence $S'$ of pairs such that each second element in a pair of $S\circ S'$ is a first element in a later pair (except for the last one), as required by the definition of CPS. 
These additional pairs do not affect the trees in $\cT$, which are already fully reduced by $S$. 
Algorithm~\ref{alg:CompletePartialSequence} describes a procedure \textsf{CompleteSeq} that runs in time linear in the length of $S$.
\begin{example}\label{ex:partialCPS}
 Let $\cT$ consist of the 2-leaf trees $(x,y)$ and $(w,z)$. A partial CPS at the end of the while loop in Algorithm~\ref{alg:CPH} could be, e.g., $S=(x,y),(w,z)$. The  trees are both reduced to one leaf, so there are no more reducible pairs, but $S$ is not a CPS. To make it into a CPS either pair $(y,z)$ or pair $(z,y)$ can be appended: e.g., $S\circ (y,z)=(x,y),(w,z),(y,z)$ is a CPS, and it still fully reduces the two input trees.
\end{example}

\begin{algorithm}[t]
 \caption{\textsf{CompleteSeq}}
  \label{alg:CompletePartialSequence}
  \begin{algorithmic}
  \Statex\textbf{INPUT:} A partial CPS $S=(x_1,y_1),\ldots,(x_n,y_n)$ that reduces $\cT$ 
  \Statex\textbf{OUTPUT:} A CPS $S'$ for $\cT$. 
  \State $C \gets \emptyset$; $P \gets \emptyset$;
  \For{$i= n,\ldots, 1$}
      \If{$y_i\not\in C$}
         \State $P \gets P\cup\{y_i\}$;
      \EndIf
      \State $C \gets C\cup\{x_i,y_i\}$;
  \EndFor
  \State $S' \gets  S$;
  \While{$|P|>1$}
     \State Let $r_1$ and $r_2$ be two arbitrary elements of $P$;
      \State $S' \gets  S'\circ(r_1,r_2)$;
      \State $P \gets  P\setminus\{r_1\}$;
  \EndWhile
  \State\Return $S'$
    \end{algorithmic}
\end{algorithm}
The class of heuristics given by Algorithm~\ref{alg:CPH} is concretised in different heuristics depending on the function $\textsf{PickNext}$ at line~\ref{line:picknext} used to choose a reducible pair at each iteration.
To formulate them we need to introduce the following notions of height pair and trivial pair.
Let $N$ be a network with branch lengths and let $(x,y)$ be a reducible pair in $N$. The \emph{height pair} of $(x,y)$ in $N$ is a pair $(h_x^N,h_y^N)\in\R_{\geq0}^2$, where $h_x^N=\ell(p(x),x)$ and $h_y^N=\ell(p(y),y)$ if $(x,y)$ is a cherry (indeed, in this case, $p(x)=p(y)$); $h_x^N=\ell(p(y),p(x))+\ell(p(x),x)$ and $h_y^N=\ell(p(y),y)$ if $(x,y)$ is a reticulated cherry. The \emph{height} $h^N_{(x,y)}$ of $(x,y)$ is the average $(h_x^N+h_y^N)/2$ of $h_x^N$ and $h_y^N$.
Let $\cT$ be a set of trees whose leaf sets are subsets of a set of taxa $X$. An ordered leaf pair $(x,y)$ is a \emph{trivial pair} of $\cT$ if it is reducible in all $T\in\cT$ that contain both $x$ and $y$, and there is at least one tree in which it is reducible.
We define the following three heuristics in the \textsc{CPH} class, resulting from as many possible implementations of $\textsf{PickNext}$.
\begin{description}[font=\normalfont]
\item[\rand.] Function $\textsf{PickNext}$ picks uniformly at random a reducible pair of $\cT_S$. 
\item[\lowpair.] Function $\textsf{PickNext}$ picks a reducible pair $(x,y)$ with the lowest average of values $h^T_{(x,y)}$ over all $T\in\cT_S$ in which $(x,y)$ is reducible (ties are broken randomly).
\item[\trivialrand.] Function $\textsf{PickNext}$ picks a trivial pair if there exists one and otherwise picks a reducible pair of $\cT_S$ uniformly at random.
\end{description}

\begin{theorem}
Algorithm~\ref{alg:CPH} computes a CPS that fully reduces $\cT$, for any function $\textsf{PickNext}$ that picks, in each iteration, a reducible pair of $\cT_S$.
\end{theorem}
\begin{proof}
The sequence $S$ is initiated as an empty sequence. Then, each iteration of the while loop (lines 2-5) of Algorithm~\ref{alg:CPH} appends one pair to $S$ that is reducible in at least one of the trees in $\cT$, and reduces it in all trees. Hence, in each iteration, the total size of $\cT_S$ is reduced, so the algorithm finishes in finite time. Moreover, at the end of the while loop, each tree in $\cT_S$ is reduced, thus the partial CPS $S$ reduces $\cT_S$.
As \textsf{CompleteSeq} only appends pairs at the end of $S$, the result of this subroutine still reduces all trees in $\cT_S$.
\end{proof}
In Section~\ref{sec:experiments} we experimentally show that \trivialrand{} produces the best results among the proposed randomised heuristics.  In the next section, we introduce a further heuristic step for \trivialrand{} which improves the output quality. 

\subsection{Improving Heuristic \trivialrand{} via Tree Expansion}\label{sub:relabelling}
Let $\cT$ be a set of trees whose leaf sets are subsets of a set of taxa $X$, let $S$ be a partial CPS for $\cT$ and let $\cT_{S}$ be the tree set obtained by reducing in order the pairs of $S$ in $\cT$.
With respect to a trivial pair $(x,y)$, each tree $T\in\cT_S$ is of one of the following types: (i) $(x,y)$ is reducible in $T$; or (ii) neither $x$ nor $y$ are leaves of $T$;  or (iii) $y$ is a leaf of $T$ but $x$ is not; or (iv) $x$ is a leaf of $T$ but $y$ is not. 

Suppose that at some iteration of  \textsf{TrivialRand}, the subroutine \textsf{PickNext} returns the trivial pair $(x,y)$. Then, before reducing $(x,y)$ in all trees, we do the following extra step: for each tree of type (iv), replace leaf $x$ with cherry $(x,y)$. We call this operation the \emph{tree expansion}: see Figure \ref{fig:relabelling_cherry}(c).
The effect of this step is that, after reducing $(x,y)$, leaf $x$ disappears from the set of trees, which would have not necessarily been the case before, because of trees of type (iv). Tree expansion followed by the reduction of $(x,y)$ can, alternatively, be seen as relabelling leaf $x$ in any tree of type (iv) by $y$. The choice of describing this relabelling as tree expansion is just for the purpose of proving Lemma~\ref{lem:ext_embed_tree}.

\begin{figure}[h]
    \centering
        \subfloat[
        ]{ \includegraphics[height=2.7cm]{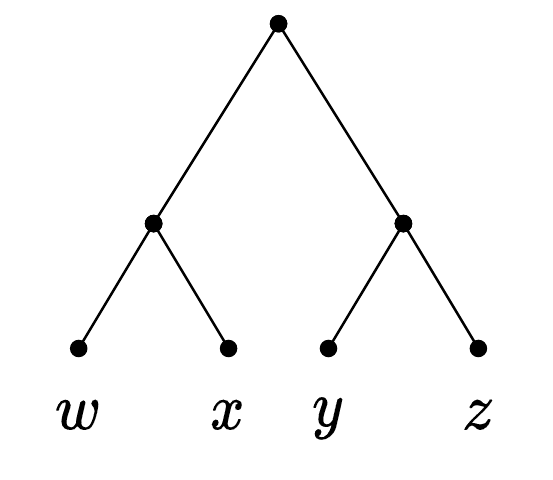} 
        \label{fig:input_tree}}
        \subfloat[
        ]{\includegraphics[height=2.7cm]{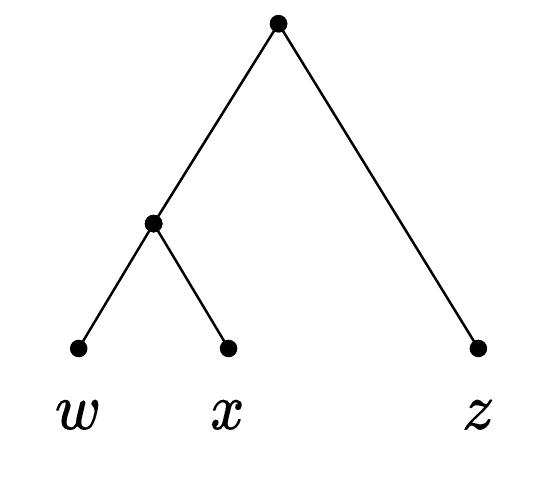} 
        \label{fig:picked_tree}}
        \subfloat[
        ]{\includegraphics[height=2.7cm]{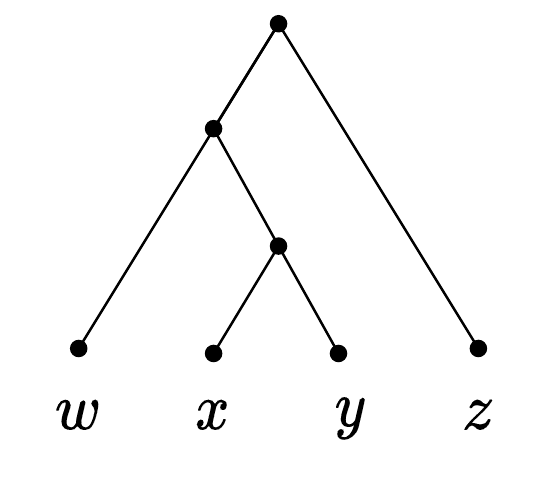} 
        \label{fig:exp_tree}}
        \subfloat[
        ]{\includegraphics[height=2.7cm]{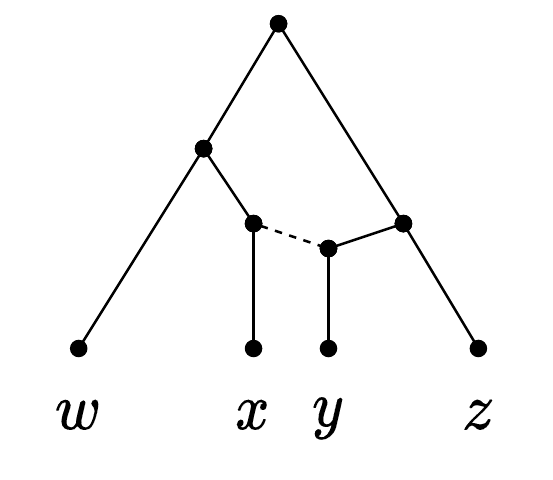}
        }
    \caption{Tree expansion of $T$  \textbf{(a)} with the trivial cherry $(x, y)$ of $\cT_{(y,z)}$. \textbf{(b)} After picking cherry $(y, z)$, leaf $y$ is missing in $T^{(1)}$. \textbf{(c)} Leaf $x$ is replaced by the cherry $(x, y)$. 
    After completion of the heuristic, we have $S_T = (y, z), (x, y), (y, w), (w, z)$. \textbf{(d)} The network $N_T$ reconstructed from $S^1\cdot(x,y)$. Note that the input tree $T$ is displayed in $N_T$ (solid edges).}
    \label{fig:relabelling_cherry}
\end{figure}

To guarantee that a CPS $S$ produced with tree expansion implies a feasible solution for \textsc{Hybridization}, we must show that the network $N$ reconstructed from $S$ displays all the trees in the input set $\cT$.
We prove that indeed this is the case with the following steps: (1), we consider the networks $N_T$ obtained by ``reverting'' a partial CPS $S$ obtained right after
applying tree expansion to a tree $T_S$: in other words, to obtain $N_T$ we add to the partially reduced tree $T_S$ the trivial pair $(x,y)$ and then all the pairs previously reduced by $S$ in the sense of Lemma~\ref{lem:binary_reconstr}. We show that $N_T$ always displays $T$, the original tree; 
(2), we prove that this holds for an arbitrary sequence of tree expansion operations; and (3), since the CPS obtained using tree expansions fully reduces the networks of point (2), and since these networks display the trees in the original set $\cT$, we have the desired property by Theorem~\ref{the:BNBRedimpliesCon}. We prove this more formally with the following lemma.

\begin{lemma}\label{lem:ext_embed_tree}
Let $S$ be the CPS produced by \textsf{TrivialRand} using tree expansion with input $\cT$. Then the network reconstructed from $S$ displays all the trees in $\cT$.
\end{lemma}
\begin{proof}
Let us start with the case where only 1 tree expansion occurs. 
Let $S^{(i-1)}$ be the partial CPS constructed in the first $i-1$ steps of \trivialrand{}, and let $i$ be the step in which we pick a trivial pair $(x,y)$. 
For each $T\in\cT_{S^{(i-1)}}$ that is reduced by $S^{(i-1)}$ to a tree $T^{(i-1)}$ of type (iv) for $(x,y)$, let $S^{(i-1)}_T$ be the subsequence of $S^{(i-1)}$ consisting only of the pairs that subsequently affect $T$. We use the partial CPS  $S^i_T=S^{(i-1)}_T\circ (x,y)$ to reconstruct a network $N_T$ with a method underlying Lemma~\ref{lem:binary_reconstr}, starting from $T^{(i-1)}$: see Figure~\ref{fig:relabelling_cherry}(d). 

For trees of type (i)-(iii), $N_T=T$.
We call the set $\mathcal{N}_\cT$, consisting of the networks $N_T$ for all $T\in\cT$, the \emph{expanded reconstruction} of $\cT$. Note that, by construction and Lemma~\ref{lem:binary_reconstr}, all the elements of $\mathcal{N}_\cT$ after reducing, in order, the pairs of $S^{(i-1)}\circ (x,y)$, are trees: in particular, they are equal to the trees of $\cT_{S^{(i-1)}\circ (x,y)}$ in which all the labels $y$ have been replaced by $x$. We denote this set of trees $(\mathcal{N}_\cT)_{S^{(i-1)}\circ (x,y)}$.

We can generalise this notion to multiple trivial pairs: we denote by $\mathcal{N}_\cT^{(j)}$ the expanded reconstruction of $\cT$ with the first $j$ trivial pairs, and suppose we added the $j$-th pair $(w,z)$ to the partial CPS $S$ at the $k$-th step. Consider a tree $T'\in(\mathcal{N}_\cT^{(j-1)})_{S^{(k-1)}}$ of type (iv) for $(w,z)$, and let $N_T^{(j-1)}\in\mathcal{N}_\cT^{(j-1)}$ be the network it originated from. 
Let $S^{(k-1)}_T$ be the subsequence of $S^{(k-1)}$ consisting only of the pairs that subsequently affected $N_T^{(j-1)}$.
Then $N_T^{(j)}$ is the network reconstructed from $S^{(k-1)}_T\circ (w,z)$, starting from $T'$.
For trees of $(\mathcal{N}_\cT^{(j-1)})_{S^{(k-1)}}$ that are of type (i)-(iii) for $(w,z)$, we have $N_T^{(j)}=N_T^{(j-1)}$.
The elements of $\mathcal{N}_\cT^{(j)}$ are all networks $N_T^{(j)}$. For completeness, we define $\mathcal{N}_\cT^{(0)}=\cT$ and $\mathcal{N}_\cT^{(1)}=\mathcal{N}_\cT$.

By construction, $S$ fully reduces all the networks in $\mathcal{N}_\cT^{(j)}$, 
thus the network $N$ reconstructed from $S$ displays all of them by Theorem~\ref{the:BNBRedimpliesCon}. 
We prove that $N_T^{(j)}$ displays $T$ for all $T\in\cT$, and thus $N$ displays the original tree set $\cT$ too, by induction on $j$.

In the base case, we pick $j=0$ trivial pairs, so the statement is true by Theorem~\ref{the:BNBRedimpliesCon}. 
Now let $j>0$. The induction hypothesis is that each network $N_T^{(j-1)}\in\mathcal{N}_\cT^{(j-1)}$ displays the tree $T\in\cT$ it originated from. Let $(w,z)$ be the $j$-th trivial pair, added to the sequence at position $k$. Let $T'\in (\mathcal{N}_\cT^{(j-1)})_{S^{(k-1)}}$ be a tree of type (iv) for $(w,z)$, and let $N_T^{(j-1)}$ be the network it originates from.
Then there are two possibilities: either $z$ is a leaf of $N_T^{(j-1)}$ or it is not. In case it is not, then adding $(w,z)$ to $N_T^{(j-1)}$ does not create any new reticulation, and clearly $N_T^{(j)}$ keeps displaying $T$. If $z$ does appear in $N_T^{(j-1)}$, then it must have been reduced by a pair $(z,v)$ of $S^{(k-1)}$ (otherwise $T'$ would not be of type (iv)). 
Then the network $N_T^{(j)}$ has an extra reticulation, created with the insertion of $(z,v)$ at some point after $(w,z)$ during the backwards reconstruction. 
In both cases, by~\cite[Lemma 10]{janssen2021cherry} $N_T^{(j-1)}$ is displayed in $N_T^{(j)}$, and thus by the induction hypothesis $T$ is displayed too.
\end{proof}

\subsection{Good Cherries in Theory}
\label{sub:good_bad_cherry}

By Lemma~\ref{lem:binary_reconstr} the binary network $N$ reconstructed from a CPS $S$ is such that $S$ is of minimum length for $N$, that is, there exists no shorter CPS that fully reduces $N$. 
By Theorem~\ref{the:BNBRedimpliesCon} if $S$, in turn, fully reduces $\cT$, then $N$ displays all the trees in $\cT$. Depending on $S$, though, $N$ is not necessarily an optimal network (i.e., with minimum reticulation number) among the ones displaying $\cT$: see Example~\ref{ex:bad_cherries}.

Let $\OPT(\cT)$ denote the set of networks that display $\cT$ with the minimum possible number of reticulations (in general, this set contains more than one network). Ideally, we would like to produce a CPS fully reducing $\cT$ that is also a minimum-length CPS fully reducing some network of $\OPT(\cT)$.
In other words, we aim to find a CPS $\tilde{S}=(x_1,y_1),\ldots,(x_n,y_n)$ such that, for any $i\in[1,n]$, $(x_i,y_i)$ is a reducible pair of $\tilde{N}_{\tilde{S}^{(i-1)}}$, where $\tilde{S}^{(0)}=\emptyset$, $\tilde{S}^{(k)}=(x_1,y_1),\ldots,(x_k,y_k)$ for all $k\in[1,n]$, and $\tilde{N}\in \OPT(\cT)$.
Let $S=(x_1,y_1),\ldots,(x_n,y_n)$ be a CPS fully reducing $\cT$ and let $\OPT^{(k)}(\cT)$ 
consist of all networks $N\in \OPT(\cT)$ such that each pair $(x_i,y_i)$, $i\in[1,k]$, is reducible in $N_{S^{(i-1)}}$. 

\begin{lemma}\label{lem:bad_cherries}
A CPS $S$ reducing $\cT$ reconstructs an optimal network $\tilde{N}$ if and only if each pair $(x_i,y_i)$ of $S$ is reducible in $\tilde{N}_{S^{i-1}}$, for all $i \in [1,n]$.
\end{lemma}
\begin{proof}
($\Rightarrow$) By Lemma~\ref{lem:binary_reconstr}, $S$ is a minimum-length CPS for the network $\tilde{N}$ that is reconstructed from it; and 
a CPS $C=(w_1,z_1),\ldots,(w_n,z_n)$ reducing a network $N$ is of minimum length precisely if,
for all $j\in[1,n]$, $(w_j,z_j)$ is a reducible pair of $N_{C^{(j-1)}}$ (otherwise the pair $(w_j,z_j)$ could be removed from $C$ and the new sequence would still reduce $N$).\\
($\Leftarrow$) If all pairs of $S$ affect some optimal network $\tilde{N}$, then $S$ is a minimum-length CPS for $\tilde{N}$, thus $\tilde{N}$ is reconstructed from $S$ (and it displays $\cT$ by Theorem~\ref{the:BNBRedimpliesCon}).
\end{proof}
Lemma~\ref{lem:bad_cherries} implies that if some pair $(x_i,y_i)$ of $S$ does not reduce any network in $\OPT^{(i-1)}(\cT)$, then the network reconstructed from $S$ is not optimal: see Example~\ref{ex:bad_cherries}.
 
\begin{example}\label{ex:bad_cherries}
Consider the set $\cT$ of Figure~\ref{fig:forest}:  $S=(y,x),(y,z),(w,x),(x,z)$ is a CPS that fully reduces $\cT$ and consists only of pairs successively reducible in the network $N$ of Fig.~\ref{fig:network}, thus it reconstructs it by Lemma~\ref{lem:binary_reconstr}. Now consider $(w, x)$, which is reducible in $\cT$ but not in $N$, and pick it as first pair, to obtain e.g. $S'=(w, x), (y, z), (y, x), (w, x), (x, z)$. The network $N'$ reconstructed from $S'$, depicted in Figure~\ref{fig:network_bad_cps}, has $r(N')=2$, whereas $r(N)=1$. 
\end{example}

\begin{figure}[h]
    \centering
    \includegraphics[height=2.7cm]{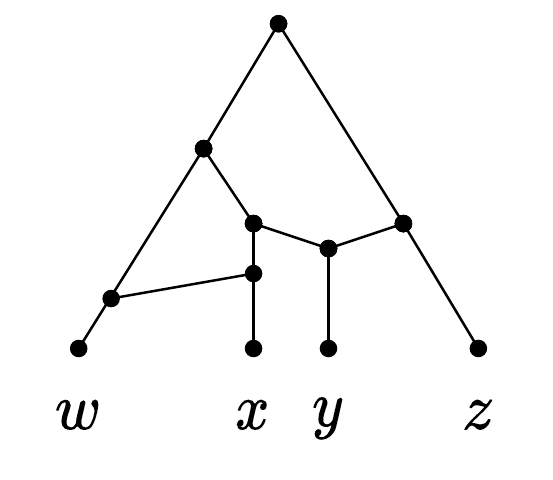}
    \caption{Network $N'$ of Example~\ref{ex:bad_cherries}.} 
    \label{fig:network_bad_cps}
\end{figure}

Suppose we are incrementally constructing a CPS $S=(x_1,y_1),\ldots,(x_n,y_n)$ for $\cT$ with some heuristic in the CPH class. 
If we had an oracle that at each iteration $i$ told us if a reducible pair $(x,y)$ of $\cT^{(i-1)}$ were a reducible pair in some $N\in\OPT^{(i-1)}(\cT)$, then, by Lemma~\ref{lem:bad_cherries}, we could solve \textsc{Hybridization} optimally.
Unfortunately no such exact oracle can exist (unless $P=NP$).
However, in the next section we exploit this idea to design machine-learned  heuristics in the \textsc{CPH} framework.

\section{Predicting Good Cherries via Machine Learning}\label{sec:ML}
In this section, we present a supervised machine-learning classifier that (imperfectly) simulates the ideal oracle described at the end of Section~\ref{sub:good_bad_cherry}. 
The goal is to predict, based on~$\cT$, whether a given cherry of~$\cT$ is a cherry or a reticulated cherry in a network~$N$ displaying $\cT$ with a close-to-optimal number of reticulations, without knowing~$N$.
Based on Lemma~\ref{lem:bad_cherries}, we then exploit the output of the classifier to define new functions \textsf{PickNext}, that in turn define new machine-learned heuristics in the class of \textsc{CPH} (Algorithm~\ref{alg:CPH}).

Specifically, we train a random forest classifier on data that encapsulates information on the cherries in the tree set.
Given a partial CPS, each reducible pair in $\cT_S$ is represented by one data point. Each data point is a pair $(\mathbf{F},\mathbf{c})$, where $\mathbf{F}$ is an array containing the features of a cherry $(x, y)$ and $\mathbf{c}$ is an array containing the probability that the cherry belongs to each of the possible classes described below. Recall that cherries are ordered pairs, so $(x,y)$ and $(y,x)$ give rise to two distinct data points.
The classification model learns the association between $\mathbf{F}$ and $\mathbf{c}$. 

The true class of a cherry $(x,y)$ of $\cT$ depends on whether, for the (unknown) network 
$N$ that we aim to reconstruct: (class 1) $(x,y)$ is a cherry of $N$; (class 2) $(x,y)$ is a reticulated cherry of $N$; (class 3) $(x,y)$ is not reducible in $N$, but $(y,x)$ is a reticulated cherry; or (class 4) neither $(x,y)$ nor $(y,x)$ are reducible in $N$.
Thus, for the data point of a cherry $(x,y)$, $\mathbf{c}[i]$ contains the probability that $(x,y)$ is in class $i$, and $\mathbf{c}[1]+\mathbf{c}[2]$ gives the predicted probability that $(x,y)$ is reducible in $N$.
We define the following two heuristics in the \textsc{CPH} framework.
\begin{description}[font=\normalfont]
\item[\ML.] Given a threshold $\tau\in[0,1)$, function $\textsf{PickNext}$ picks the cherry with the highest predicted probability of being reducible in $N$ if this probability is at least $\tau$; 
or a random cherry if none has a probability of being reducible above $\tau$.
\item[\trivML.] Function $\textsf{PickNext}$ picks a random trivial pair, if there exists one; otherwise it uses the same rules as \ML. 
\end{description}

In both cases, whenever a trivial pair is picked, we do tree expansion, as described in Section~\ref{sub:relabelling}. Note that if $\tau=0$, since the predicted probabilities are never exactly 0, \ML{} is fully deterministic. In Section~\ref{exp:threshold} we show how the performance of \ML{} is impacted by the choice of different thresholds.
\begin{table}[H]
	\caption{Features of a cherry $(x, y)$. Features 6-12 can be computed for both branch lengths and unweighted branches.
	We refer to these two options as \emph{distance} and \emph{topological distance}, respectively. }
	\centering
	\footnotesize
    \begin{tabular}{lll}
		\toprule
		{\scriptsize Num} & {\scriptsize Feature name} & {\scriptsize Description}	\\
		\midrule
		{\scriptsize 1}   & {\scriptsize Cherry in tree}           &  {\scriptsize Ratio of trees that contain cherry $(x,y)$}\\
	    {\scriptsize 2 }  & {\scriptsize New cherries}       & {\scriptsize  Number of new cherries of $\cT$ after picking cherry $(x, y)$ }                 \\
		{\scriptsize 3}   & {\scriptsize Before/after}      &  {\scriptsize Ratio of the number of cherries of $\cT$ before/after picking cherry $(x, y)$ }  \\ 
		{\scriptsize 4}   & {\scriptsize Trivial}                   &  {\scriptsize Ratio of trees with both leaves $x$ and $y$ that contain cherry $(x,y)$}   \\
		{\scriptsize 5}   & {\scriptsize Leaves in tree}         &  {\scriptsize Ratio of trees that contain both leaves $x$ and $y$}                              \\
        \midrule
        \multicolumn{3}{l}{\textit{\scriptsize{Features measured by distance (d) and topology (t)}}} \\
        \midrule
        {\scriptsize $6_{d,t}$}   & {\scriptsize Tree depth}               &  {\scriptsize Avg over trees with $(x,y)$ of ratios ``depth of the tree/max depth over all trees'' }                                   \\
		{\scriptsize$7_{d,t}$}   & {\scriptsize Cherry depth }            & {\scriptsize Avg over trees with $(x,y)$ of ratios ``depth of $(x,y)$ in the tree/depth of the tree'' }       \\
        {\scriptsize $8_{d,t}$}   & {\scriptsize Leaf distance }            &   {\scriptsize Avg over trees with $x$ and $y$ of ratios ``$x$-$y$ leaf distance/depth of the tree''}\\ 
        {\scriptsize $9_{d,t}$}  & {\scriptsize Leaf depth $x$}           & {\scriptsize Avg over trees with $x$ and $y$ of ratios ``root-$x$ distance/depth of the tree''}                                                   \\
       {\scriptsize $10_{d,t}$}  & {\scriptsize Leaf depth $y$  }         & {\scriptsize Avg over trees with $x$ and $y$ of ratios ``root-$y$ distance/depth of the tree''}     \\
       {\scriptsize $11_{d,t}$}   & {\scriptsize LCA distance }      & {\scriptsize Avg over trees with $x$ and $y$ of ratios ``$x$-LCA$(x,y)$ distance/$y$-LCA$(x,y)$ distance''}\\
       {\scriptsize $12_{d,t}$ } & {\scriptsize Depth $x$/$y$ }   & {\scriptsize Avg over trees with $x$ and $y$ of ratios ``root-$x$ distance/root-$y$ distance''}\\
		\bottomrule
	\end{tabular}
	\label{tab:cherry_features}
\end{table}
To assign a class to each cherry, we define 19 features, summarised in Table~\ref{tab:cherry_features}, that may capture essential information about the structure of the set of trees, and that can be efficiently computed and updated at every iteration of the heuristics. 

The \emph{depth} (resp. \emph{topological} depth) of a node $u$ in a tree $T$ is the total branch length (resp. the total number of edges) on the root-to-$u$ path; the depth of a cherry $(x,y)$ is the depth of the common parent of $x$ and $y$; the depth of $T$ is the maximum depth of any cherry of $T$.
The (topological) leaf distance between $x$ and $y$ is the total branch length of the path from the parent of $x$ to the lowest common ancestor of $x$ and $y$, denoted by LCA$(x,y)$, plus the total length of the path from the parent of $y$ to LCA$(x,y)$ (resp. the total number of edges on both paths). In particular, the leaf distance between the leaves of a cherry is zero. 

\subsection{Time Complexity}\label{sub:time_complexity}

Designing algorithms with the best possible time complexity was not the main objective of this work. However, for completeness, we provide worst-case upper bounds on the running time of our heuristics. The omitted proofs can be found in Appendix~\ref{app:time_complex}. We start by stating a general upper bound for the whole CPH framework in the function of the time required by the \textsf{PickNext} routine.

\begin{lemma}\label{lem:time_rand}
The running time of the heuristics in the CPH framework is $\cO(|\cT|^2 |X|+cost(\textsf{PickNext}))$, where $cost(\textsf{PickNext})$ is the total time required to choose reducible pairs over all iterations. In particular, \rand{} takes $\cO(|\cT|^2 |X|)$ time.
\end{lemma}

\begin{proof}
An upper bound for the sequence length is $(|X|-1)|\cT|$ as each tree can individually be fully reduced using at most $|X|-1$ pairs. Hence, the while loop of Algorithm~\ref{alg:CPH} is executed at most $(|X|-1)|\cT|$ times. Moreover, reducing the pair and updating the set of reducible pairs after one iteration takes $O(1)$ time per tree. Combining this with the fact that \textsf{CompleteSeq} takes $\cO(|S|)=\cO(|X||\cT|)$ time, we obtain the stated time complexity.
Since choosing a random reducible pair takes $\cO(1)$ time at each iteration, \rand{} takes trivially $\cO(|\cT|^2 |X|)$ time.
\end{proof}

Note that, by Lemma~\ref{lem:CPS-retic}, the number of reticulations $r(N)$ of the network reconstructed from the output CPS is bounded by $(|X|-1)|\cT|-|X|+1$ and thus the time complexity of \rand{} is also $\cO(r(N)|\cT|)$.

Let us now focus on the time complexity of the machine-learned heuristics \ML{} and \trivML{}.
At any moment during the execution of the heuristics, we maintain a data structure that stores all the current cherries in $\cT$ and allows constant-time insertions, deletions, and access to the cherries and their features. 
A possible implementation of this data structure consists of a hashtable $\cherries$ paired with a list \cherrylist{} of the pairs currently stored in \cherries.
We will use $\cherrylist{}$ to iterate over the current cherries of $\cT$, and $\cherries{}$ to check whether a certain pair is currently a cherry of $\cT$ and to access its features.

Note that the total number of cherries inserted in $\cherries$ over all the iterations is bounded by the total size of the trees $||\cT||$ because up to two cherries can be created for each internal node over the whole execution. 
We will assume that we have constant-time access to the leaves of each tree: specifically, given $T\in\cT$ and $x\in X$, we can check in constant time whether $x$ is currently a leaf of $T$\footnote{This can be obtained maintaining a list of leaves of each tree and a hashtable with the leaves as keys: the value of a key $x$ is a pointer to the position of $x$ in the list.}.

\paragraph{Initialisation}
The cherries of $\cT$ can be identified and features 1-3 can be initially computed in $\cO(||\cT||)$ time by traversing all trees bottom-up. Features 4-5 can be computed in $\cO(\min\{|\cT|\cdot ||\cT||,|\cT|\cdot|X|^2\})$ time by checking, for each $T\in\cT$ and each cherry $(x,y)$ of $\cT$, whether both $x$ and $y$ appear in $T$.  
Features $6_{d,t}$ to $12_{d,t}$ can also be initially computed with a traversal of $\cT$ made efficient by preprocessing each tree in linear time to allow constant-time LCA queries~\cite{DBLP:journals/siamcomp/HarelT84} and by storing the depth (both topological and with the branch lengths) of each node. We also store the topological and branch length depth of each tree and their maximum value over $\cT$.  Altogether this gives the following lemma.

\begin{lemma}\label{lem:initialization}
Initialising all features for a tree set $\cT$ of total size $||\cT||$ over a set of taxa $X$ requires  $\cO(\min\{|\cT|\cdot ||\cT||,|\cT|\cdot|X|^2\})$ time and $\cO(||\cT||)$ space.
\end{lemma}
The next lemma provides an upper bound on the time complexity of updating the distance-independent features. 

\begin{restatable}{lemma}{lemfeaturesA}\label{lem:feaures1-5}
Updating features 1-5 for a set $\cT$ of $|\cT|$ trees of total size $||\cT||$ over a set of taxa $X$ requires $\cO(|\cT|(||\cT||+|X|^2))$ total time and $\cO(||\cT||)$ space.
\end{restatable}
 Since searching for trivial cherries at each iteration of the randomised heuristic \trivialrand{} can be done with the same procedure we use for updating feature $4$ in the machine-learned heuristics, which in particular requires $\cO(|\cT|\cdot||\cT||)$ time, we have the following corollary.

\begin{corollary}
The time complexity of \trivialrand{} is $\cO(|\cT|\cdot||\cT||)=\cO(|\cT|^2\cdot |X|)$.
\end{corollary}

The total time required for updating the distance-dependent features raises the time complexity of \ML{} and \trivML{} to quadratic in the input size. However, the extensive analysis reported in Appendix~\ref{app:time_complex} shows that this is only due to the single feature $6_d$, and without such a feature, the machine-learned heuristics would be asymptotically as fast as the randomised ones. Since Table~\ref{tab:feature_importances} in Appendix~\ref{app:Ml_models_heatmaps} shows that this feature is not particularly important, in future work it could be worth investigating whether disregarding it leads to equally good results in shorter time.
\begin{restatable}{lemma}{lemfeaturesB}\label{lem:ML_running_time}
    The time complexity of \ML{} and \trivML{} is $\cO(||\cT||^2)$.
\end{restatable}

\subsection{Obtaining Training Data}\label{sub:training}
The high-level idea to obtain training data is to first generate a phylogenetic network $N$; then to extract the set $\cT$ of all the exhaustive trees displayed in $N$; and finally, to iteratively choose a random reducible pair $(x,y)$ of $N$, to reduce it in $\cT$ as well as in~$N$, and to label the remaining cherries of $\cT$ with one of the four classes defined in Section~\ref{sec:ML} until the network is fully reduced.

We generate two different kinds of binary orchard networks, normal and not normal, with branch lengths and up to 9 reticulations using the LGT (lateral gene transfer) network generator of~\cite{pons2019generation}, imposing normality constraints when generating the normal networks.
For each such network $N$, we then generate the set $\cT$ consisting of all the exhaustive trees displayed in $N$. 

If $N$ is normal, $N$ is an optimal network for $\cT$~\cite[Theorem 3.1]{willson2010regular}. 
This is not necessarily true for any LGT-generated network, but even in this case, we expect $N$ to be reasonably close to optimal, because we remove redundant reticulations when we generate it and because the trees in $\cT$ cover all the edges of $N$. In particular, for LGT networks $r(N)$ provides an upper bound estimate on the minimum possible number of reticulations of any network displaying~$\cT$, and we will use it as a reference value for assessing the quality of our results on synthetic LGT-generated data.

\section{Experiments}\label{sec:experiments}

The code of all our heuristics and for generating data is written in Python and is available at \texttt{https://github.com/estherjulien/learn2cherrypick}. All experiments ran on an Intel Xeon Gold 6130 CPU @ 2.1 GHz with 96 GB RAM. We conducted experiments on both synthetic and real data, comparing the performance of \rand{}, \trivialrand{}, \ML{} and \trivML{}, using threshold $\tau=0$. 
Similar to the training data, we generated two synthetic datasets by first growing a binary orchard network $N$ using~\cite{pons2019generation}, and then extracting $\cT$ as a subset of the exhaustive trees displayed in $N$. We provide details on each dataset in Section~\ref{sub:heu_results}. 

We start by analysing the usefulness of tree expansion, the heuristic rule described in Section~\ref{sub:relabelling}. We synthetically generated 112 instances for each tree set size $|\cT|\in\{5,10,20,50,100\}$ (560 in total), all consisting of trees with 20 leaves each, and grouped them by $|\cT|$; we then ran \trivialrand{} 200 times (both with and without tree expansion) on each instance, selected the best output for each of them, and finally took the average of these results over each group of instances. The results are in Figure~\ref{fig:trivial}, showing that the use of tree expansion brought the output reticulation number down by at least 16\% (for small instances) and up to 40\% for the larger instances. We consistently chose to use this rule in all the heuristics that detect trivial cherries, namely, \trivialrand{}, \trivML{},
\ML{} (although \ML{} does not explicitly favour trivial cherries, it does check whether a selected cherry is trivial using feature number 2), and the non-learned heuristic that will be introduced in Section~\ref{sec:new_heuristic}.

\begin{figure}[t]
    \centering
   \includegraphics[width=0.8
   \columnwidth]{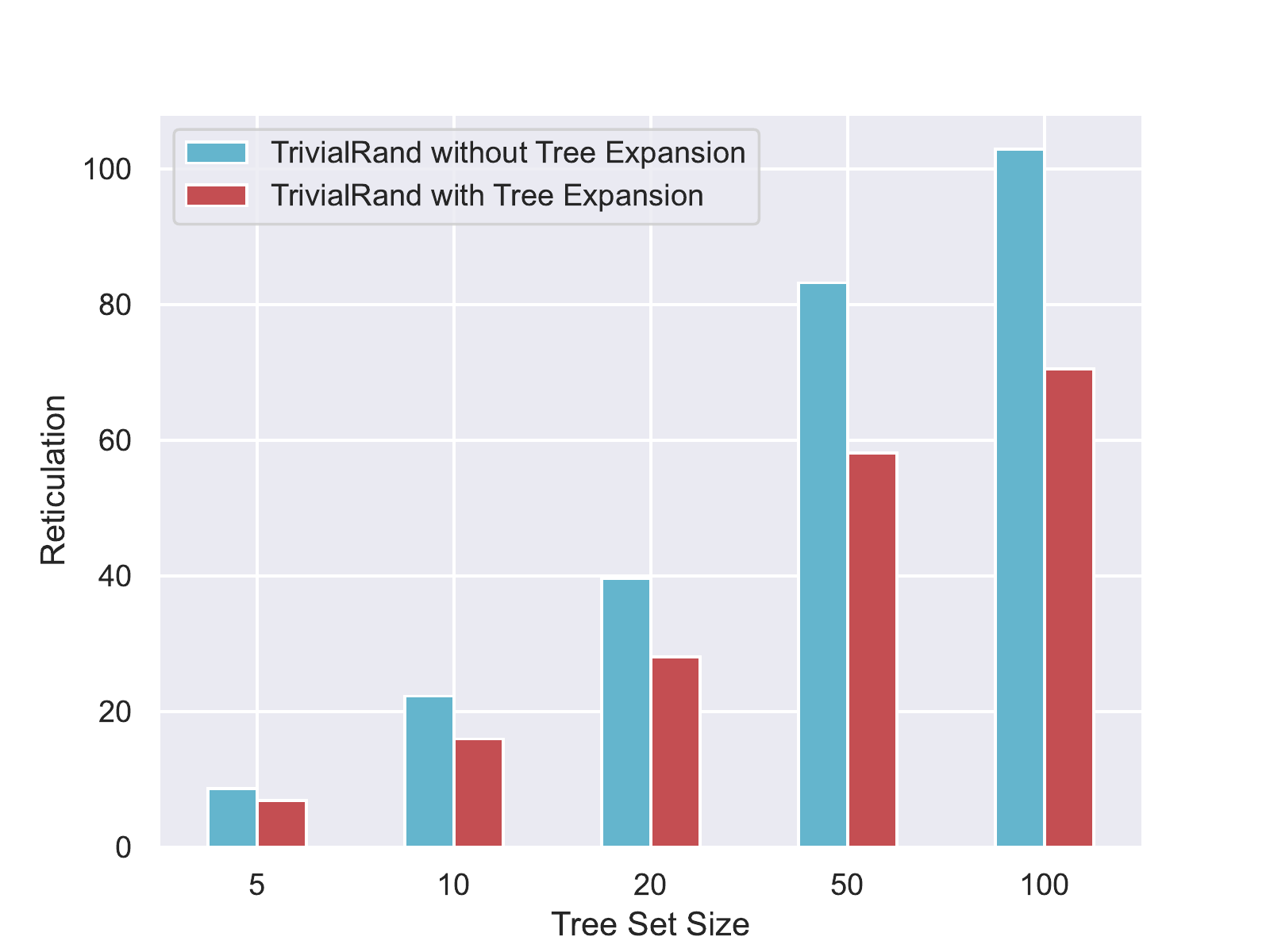}
    \caption{Number of reticulations output by \trivialrand{} with and without using tree expansion. The height of the bars is the average reticulation number over each group, obtained by selecting the best of 200 runs for each instance.}
    \label{fig:trivial}
\end{figure}

\subsection{Prediction Model}
The random forest is implemented with Python's \texttt{scikit-learn} \cite{scikit-learn} package using default settings. We evaluated the performance of our trained random forest models on different datasets in a holdout procedure: namely, we removed 10\% of the data from each training dataset, trained the models on the remaining 90\% and used the holdout 10\% for testing. The accuracy was assessed by assigning to each test data point the class with the highest predicted probability and comparing it with the true class. Before training the models, we balanced each dataset so that each class had the same number of representatives.

Each training dataset differed in terms of the number $M$ of networks used for generating it and the number of leaves of the networks. For each dataset, the number $L$ of leaves of each generated network was uniformly sampled from $[2, \max L]$, 
where $\max L$ is the maximum number of leaves per network. We constructed LGT networks using the LGT generator of~\cite{pons2019generation}. This generator has three parameters: $n$ for the number of steps, $\alpha$  for the probability of lateral gene transfer events, and $\beta$ for regulating the size of the biconnected components of the network (called \emph{blobs}). The combination of these parameters determines the level (maximum number of reticulations per blob), the number of
reticulations, and the number of leaves of the output network. In our experiments, $\alpha$ was uniformly sampled from $[0.1, 0.5]$ and $\beta=1$ (see \cite{pons2019generation} for more details).

To generate normal networks we used the same generator with the same parameters, but before adding a reticulation we check if it respects the normality constraints and only add it if it does. Each generated network gave rise to a number of data points: the total number of data points per dataset is shown in Table \ref{tab:ml_description} in Appendix~\ref{app:random_forests}.
Each row of Table \ref{tab:ml_description} corresponds to a dataset on which the random forest can be trained, obtaining as many ML models. We tested all the models on all the synthetically generated instances: we show these results in Figures~\ref{fig:ml_heur_perf_norm},~\ref{fig:ml_heur_perf_LGT} and~\ref{fig:ml_heur_perf_ZODS} in Appendix~\ref{app:Ml_models_heatmaps}. In Section~\ref{sub:heu_results} we will report the results obtained for the best-performing model for each type of instance. 

Among the advantages of using a random forest as a prediction model, there is the ability of computing feature importance, shown in Table \ref{tab:feature_importances} in Appendix \ref{app:random_forests}. Some of the most useful features for a cherry $(x,y)$  appear to be `Trivial' (the ratio of the trees containing both leaves $x$ and $y$ in which $(x,y)$ is a cherry) and `Cherry in tree' (the ratio of trees that contain $(x,y)$). 
This was not unexpected, as these features are well-suited to identify trivial cherries. 

`Leaf distance' (t,d), `LCA distance' (t) and `Depth $x/y$' (t) are also important features. The rationale behind these features was to try to identify reticulated cherries. This was also the idea for the feature `Before/after', but this has, surprisingly, a very low importance score. 
In future work, we plan to conduct a thorough analysis of whether some of the seemingly least important features can be removed without affecting the quality of the results.
 
\subsection{Experimental Results}\label{sub:heu_results}
We assessed the performance of our heuristics on instances of four types: 
normal, LGT, ZODS (binary non-orchard networks), and real data. Normal, LGT and ZODS data are synthetically generated. We generated the normal instances much as we did for the training data: we first grew a normal network using the LGT generator and then extracted all the exhaustive trees displayed in the network. We generated normal data for different combinations of the following parameters: $L \in \{20, 50, 100\}$ (number of leaves per tree) and $R \in \{5, 6, 7\}$ (reticulation number of the original network). Note that, for normal instances, $|\cT| = 2^R$. For every combination of the parameters $L$ and $R$ we generated 48 instances: by \emph{instance group} we indicate the set of instances generated for one specific parameter pair.

For the LGT instances, we grew the networks using the LGT generator, but unlike for the normal instances we then extracted only a subset of the exhaustive trees from each of them, up to a certain amount $|\cT| \in \{20, 50, 100\}$. The other parameters for LGT instances are the number of leaves $L \in \{20, 50, 100\}$ and the number of reticulations $R \in \{10, 20, 30\}$. 
For a fixed pair $(L,|\cT|)$, we generated 16 instances for each possible value of $R$, and analogously, for a fixed pair $(L,R)$ we generated 16 instances for each value of $|\cT|$. The 48 instances generated for a fixed pair of values constitute a LGT instance group.

We generated non-orchard binary networks using the ZODS generator~\cite{zhang2018bayesian}. This generator has two user-defined parameters: $\lambda$, which regulates the speciation rate, and $\nu$, which regulates the hybridization rate. Following~\cite{janssen2021comparing} we set $\lambda = 1$ and we sampled $\nu\in[0.0001, 0.4]$ uniformly at random. Like for the LGT instances, we generated an instance group of size $48$ for each pair of values $(L,|\cT|)$ and $(L,R)$, with $L \in \{20, 50, 100\}$, $|\cT| \in \{20, 50, 100\}$, $R \in \{10, 20, 30\}$.

Finally, the real-world dataset consists of gene trees on homologous gene sets found in bacterial and archaeal genomes, was originally constructed in \cite{beiko2011telling} and made binary in \cite{van2019practical}. We extracted a subset of instances (Table~\ref{tab:real_data}) 
from the binary dataset, for every combination of parameters $L\in\{20, 50, 100\}$ and $|\mathcal{T}|\in\{10,20,50,100\}$.

\begin{table}
    \centering
        \caption{Number of real data instances for each group (combination of parameters $L$ and $|\cT|$).}
        \begin{tabular}{llccccll}
        \toprule
            {\scriptsize$L$} & & \multicolumn{4}{c}{{\scriptsize$|\cT|$}} & & {\scriptsize Tot. Trees}\\ \hhline{-~----~-}
        &  & {\scriptsize10}  &  {\scriptsize20}  &  {\scriptsize50}  &  {\scriptsize100} &  &\\
        \midrule
        {\scriptsize 20}  &  & {\scriptsize 50} &   {\scriptsize 50} &   {\scriptsize 50} &  {\scriptsize 50} & & {\scriptsize 1684}\\
        {\scriptsize 50}  &  & {\scriptsize 20} &   {\scriptsize 20} &   {\scriptsize 20} &  {\scriptsize 20} & & {\scriptsize 290}\\
        {\scriptsize 100} &  &  {\scriptsize 5} &    {\scriptsize 5} &    {\scriptsize 1} &    {\scriptsize 0} & & {\scriptsize 53}\\
        \bottomrule
        \end{tabular}

    \label{tab:real_data}
\end{table}

For the synthetically generated datasets, we evaluated the performance of each heuristic in terms of the output number of reticulations, comparing it with the number of reticulations of the network $N$ from which we extracted $\cT$. 
For the normal instances, $N$ is the optimal network~\cite[Theorem 3.1]{willson2010regular}; this is not true, in general, for the LGT and ZODS datasets, but even in these cases, $r(N)$ clearly provides an estimate (from above) 
of the optimal value, and thus we used it as a reference value for our experimental evaluation.

For real data, in the absence of the natural estimate on the optimal number of reticulations provided by the starting network, we evaluated the performance of the heuristics comparing our results with the ones given by the exact algorithms from~\cite{van2019practical} (\Treechild) and from~\cite{hybroscale} (\Hybro), using the same datasets that were used to test the two methods in \cite{van2019practical}. These datasets consist of rather small instances ($|\cT|\leq 8$); for larger instances, we run \trivialrand{} 1000 times for each instance group, selected the best result for each group, and used it as a reference value (Figure~\ref{fig:real_instances}).

We now describe in detail the results we obtained for each type of data and each of the algorithms we tested.

\subsubsection{Experiments on Normal Data}\label{exp:normal}
For the experiments in this section we used he ML model trained on 1000 normal networks with at most 100 leaves per network (see Figure~\ref{fig:ml_heur_perf_norm} in Appendix~\ref{app:Ml_models_heatmaps}). We ran the machine-learned heuristics once for each instance and then averaged the results within each instance group (recall that one instance group consists of the sets of all the exhaustive trees of 48 normal networks having the same fixed number of leaves and reticulations). The randomised heuristics \rand{} and \trivialrand{} were run $\min\{x(I), 1000\}$ times for each instance $I$, where $x(I)$ is the number of runs that can be executed in the same time as one run of \ML{} on the same instance. We omitted the results for \textsf{LowPair} because they were at least 44\% worse on average than the worst-performing heuristic we report.

\begin{figure}[t]
    \centering
   \includegraphics[width=\columnwidth]{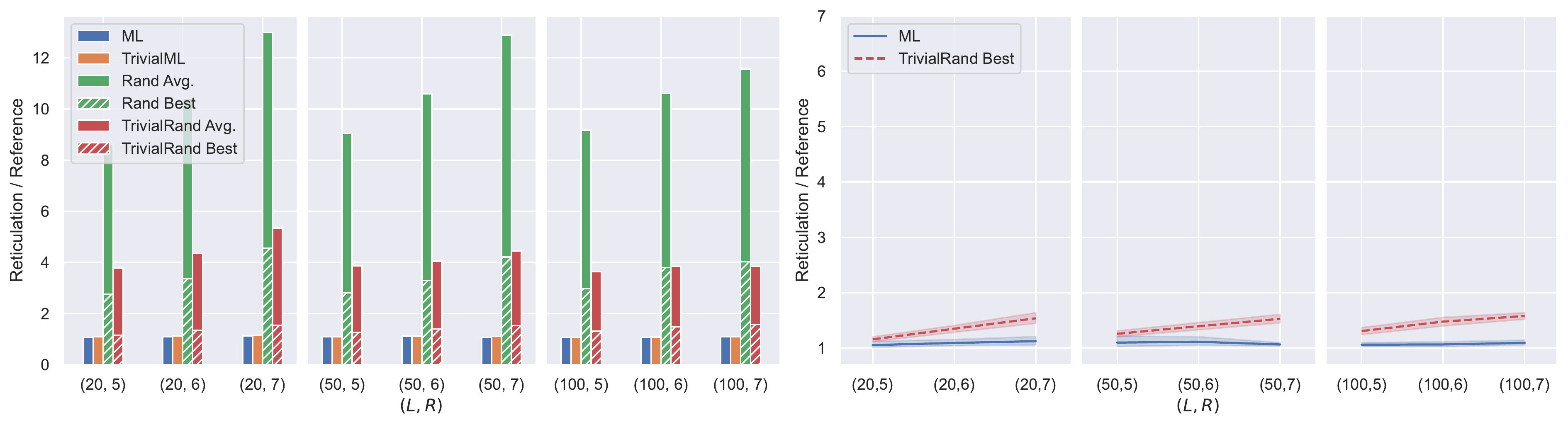}
    \caption{Experimental results for normal data. Each point on the horizontal axis corresponds to one instance group. In the left graph, the height of each bar gives the average of the results over all instances of the group, scaled by the optimum value for the group. The right graph compares the average output of \ML{} within each instance group and the average of the best output given by \trivialrand{} for each instance of a group. The shaded areas represent 95\% confidence intervals.}
    \label{fig:results_normal}
\end{figure}

In Figure~\ref{fig:results_normal} we summarise the results. 
Solid bars represent the ratio between the average reported reticulation number and the optimal value, for each instance group and for each of the four heuristics. 
Dashed bars represent the ratio between the average (over the instances within each group) of the best result among the $\min\{x(I), 1000\}$ runs for each instance $I$ and the optimum. 

The machine-learned heuristics \ML{} and \trivML{} seem to perform very similarly, both leading to solutions close to optimum. The average performance of 
\trivialrand{} is around 4 times worse than the machine-learned heuristics; in contrast, if we only consider the best solution among the multiple runs for each instance, they are quite good, having only up to 49\% more reticulations than the optimal solution, but they are still at least 4\% worse (29\% worse on average) than the machine-learned heuristics' solutions: see the right graph of igure~\ref{fig:results_normal}.

The left graph of Figure~\ref{fig:results_normal} shows that the performance of the randomised heuristics seems to be negatively impacted by the number of reticulations of the optimal solution, 
while we do not observe a clear trend for the machine-learned heuristics, whose performance is very close to optimum for all the considered instance groups. Indeed, the number of existing phylogenetic networks with a certain number of leaves grows exponentially in the number of reticulations, thus making it less probable to reconstruct a ``good" network with random choices. This is consistent with the existing exact methods being FPT in the number of reticulations~\cite{DBLP:journals/siamcomp/WhiddenBZ13,van2019practical}.

The fully randomised heuristic \rand{} always performed much worse than all the others, indicating that identifying the trivial cherries has a great impact on the effectiveness of the algorithms (recall that \ML{} implicitly identifies trivial cherries).

\subsubsection{Experiments on LGT Data}\label{exp:LGT} For the experiments on LGT data we used the ML model trained on 1000 LGT networks with at most 100 leaves per network (see Figure~\ref{fig:ml_heur_perf_LGT} in Appendix~\ref{app:Ml_models_heatmaps}). The setting of the experiments is the same as for the normal data (we run the randomised heuristics multiple times and the machine-learned heuristics only once for each instance), with two important differences. 
\begin{figure}[t]
    \centering
   \includegraphics[width=\columnwidth]{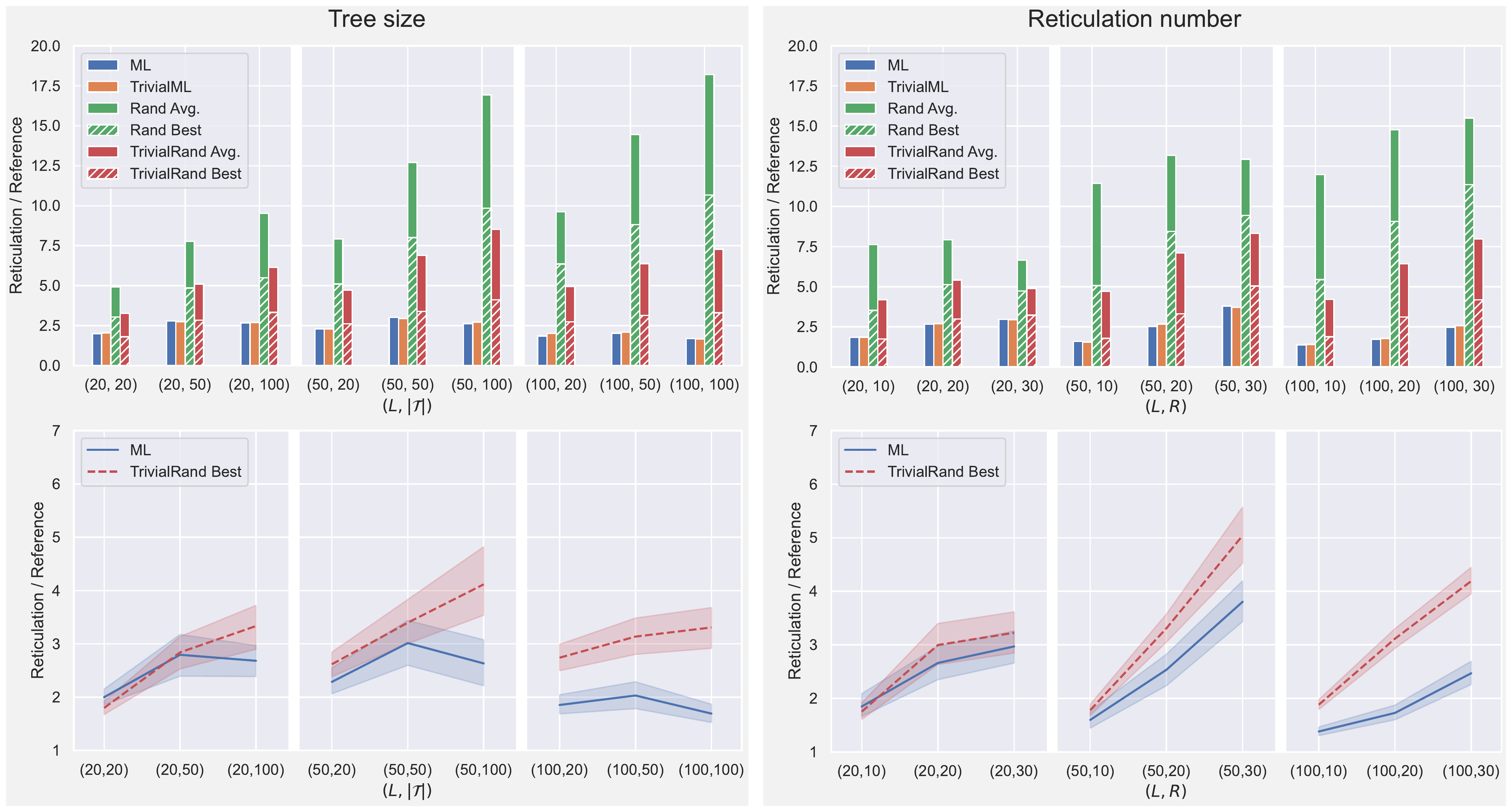}
    \caption{Experimental results for LGT data. Each point on the horizontal axis corresponds to one instance group. For the graphs on the left, there is one group for each fixed pair $(L,|\cT|)$ consisting of 16 instances coming from LGT networks for each value of $R\in\{10,20,30\}$. For the graphs on the right, there is one group for each fixed pair $(L,R)$ consisting of 16 instances coming from LGT networks for each value of $|\cT|\in\{20,50,100\}$. In the top graphs, the height of each bar gives the average of the results over all instances of the group, each scaled by the number of reticulations of the generating network. The bottom graphs compare the average output of \ML{} within each instance group and the average of the best output given by \trivialrand{} for each instance group. The shaded areas represent 95\% confidence intervals.}
    \label{fig:results_LGT}
\end{figure}
First, for LGT data we only take proper subsets of the exhaustive trees displayed by the generating  networks, and thus we have two kinds of instance groups: one where in each group the number of trees extracted from a network and the number of leaves of the networks are fixed, but the trees come from networks with different numbers of reticulations; and one where the number of reticulations of the generating networks and their number of leaves  are fixed, but 
the number of trees extracted from a network varies. 

The second important difference is that the reference value we use for LGT networks is not necessarily the optimum, but it is just an upper bound given by the number of reticulations of the generating networks which we expect to be reasonably close to the optimum (see Section~\ref{sub:training}).

The results for the LGT datasets are shown in Figure~\ref{fig:results_LGT}. Comparing these results with those of Figure~\ref{fig:results_normal}, it is evident that the LGT instances were more difficult than the normal ones for all the tested heuristics: this could be due to the fact that the normal instances consisted of all the exhaustive trees of the generating networks, while the LGT instances only have a subset of them and thus carry less information.

The machine-learned heuristics performed substantially better (up to 80\% on average) than the best randomised heuristic \trivialrand{} in all instance groups but the ones with the smallest values for parameters $R,|\cT|$ and $L$, for which the performances are essentially overlapping. On the contrary, the advantage of the machine-learned methods is more pronounced when the parameters are set to the highest values. This is because the larger the parameters, the more the possible different networks that embed $\cT$, thus the less likely for the randomised methods to find a good solution.

From the graphs on the right of Figure~\ref{fig:results_LGT}, it seems that the number of reticulations has a negative impact on both machine-learned and randomised heuristics, the effect being more pronounced for the randomised ones. The effect of the number of trees $|\cT|$ on the quality of the solutions is not as clear (Figure~\ref{fig:results_LGT}, left). However, we can still see that the trend of \ML{} and \trivialrand{} is the same: the ``difficult" instance groups are so for both heuristics, even if the degradation in the quality of the solutions for such instance groups is less marked for \ML{} than for \trivialrand{}.

\subsubsection{Experiments on ZODS Data}\label{exp:ZODS}
For the experiments on ZODS data we used the ML model trained on 1000 LGT networks with at most 100 leaves per network (see Figure~\ref{fig:ml_heur_perf_ZODS} in Appendix~\ref{app:Ml_models_heatmaps}). The setting of the experiments is the same as for the LGT data, and the results are shown in Figure~\ref{fig:results_ZODS}.
\begin{figure}[t]
    \centering
    \includegraphics[width=\columnwidth]{
    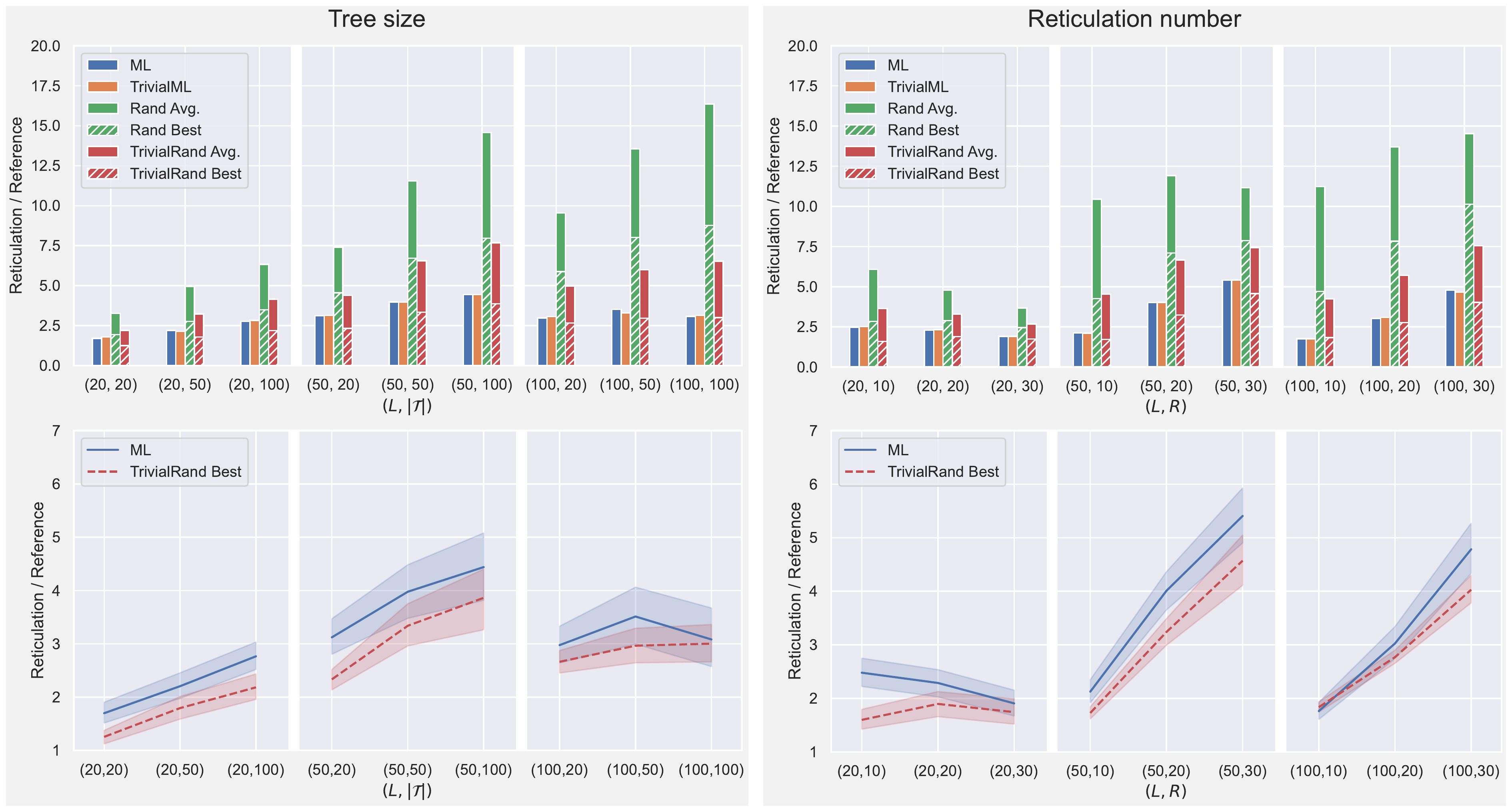}
    \caption{Experimental results for ZODS data. Each point on the horizontal axis corresponds to one instance group. For the graphs on the left, there is one group for each fixed pair $(L,|\cT|)$ consisting of 16 instances coming from ZODS networks for each value of $R\in\{10,20,30\}$. For the graphs on the right, there is one group for each fixed pair $(L,R)$ consisting of 16 instances coming from ZODS networks for each value of $|\cT|\in\{20,50,100\}$. In the top graphs, the height of each bar gives the average of the results it represents over all instances of the group, each scaled by the number of reticulations of the network the instance originated from. The bottom graphs compare the average output of \ML{} within each instance group and the average of the best output given by \trivialrand{} for each group instance. The shaded areas represent 95\% confidence intervals.}
    \label{fig:results_ZODS}
\end{figure}

At first glance, the performance of the randomised heuristics seems to be better for ZODS data than for LGT data (compare figures~\ref{fig:results_LGT} and~\ref{fig:results_ZODS}), which sounds counterintuitive. 
Recall, however, that all the graphs show the ratio between the number of reticulations returned by our methods and a reference value, i.e., the number of reticulations of the generating network: while we expect this reference to be reasonably close to the optimum for LGT networks, this is not the case for ZODS networks. 
In fact, a closer look to ZODS networks shows that they have a large number of redundant reticulations which could be removed without changing the set of trees they display, and thus their reticulation number is in general quite larger than the optimum. This is an inherent effect of the ZODS generator not having any constraints on the reticulations that can be introduced, and it is more marked on networks with a small number of leaves.

Having a reference value significantly larger than the optimum makes the ratios shown in Figure~\ref{fig:results_ZODS} small (close to 1, especially for \trivialrand{} on small instances) without implying that the results for the ZODS data are better than the ones for the LGT data.
The graphs of Figures~\ref{fig:results_LGT} and~\ref{fig:results_ZODS} are thus not directly comparable.

The reference value for the experiments on ZODS data not being realistically close to the optimum, however, does not invalidate their significance. Indeed,
the scope of such experiments was just to compare the performance of the machine-learned heuristics on data entirely different from those they were trained on with the performance of the randomised heuristics, which should not depend on the type of network that was used to generate the input.

As expected and in contrast with normal and LGT data, the results show that the machine-learned heuristics perform worse than the randomised ones on ZODS data, consistent with the ML methods being trained on a completely different class of networks. 

\subsubsection{Experiments on Real Data}\label{exp:real}
We conducted two sets of experiments on real data, using the ML model trained on the dataset trained on $1000$ LGT networks with at most 100 leaves each.
For sufficiently small instances, we compared the results of our heuristics with the results of two existing tools for reconstructing networks from binary trees: \Treechild~\cite{van2019practical} and \Hybro~\cite{hybroscale}. \Hybro{} is an exact method performing an exhaustive search on the networks displaying the input trees, therefore it can only handle reasonably small instances in terms of the number of input trees. \Treechild{} is a fixed-parameter (in the number of reticulations of the output) exact algorithm that reconstructs the best \emph{tree-child} network, a restricted class of phylogenetic networks, and due to its fast-growing computation time cannot handle large instances either.

\begin{figure}[t]
    \centering
    \includegraphics[width=\columnwidth]{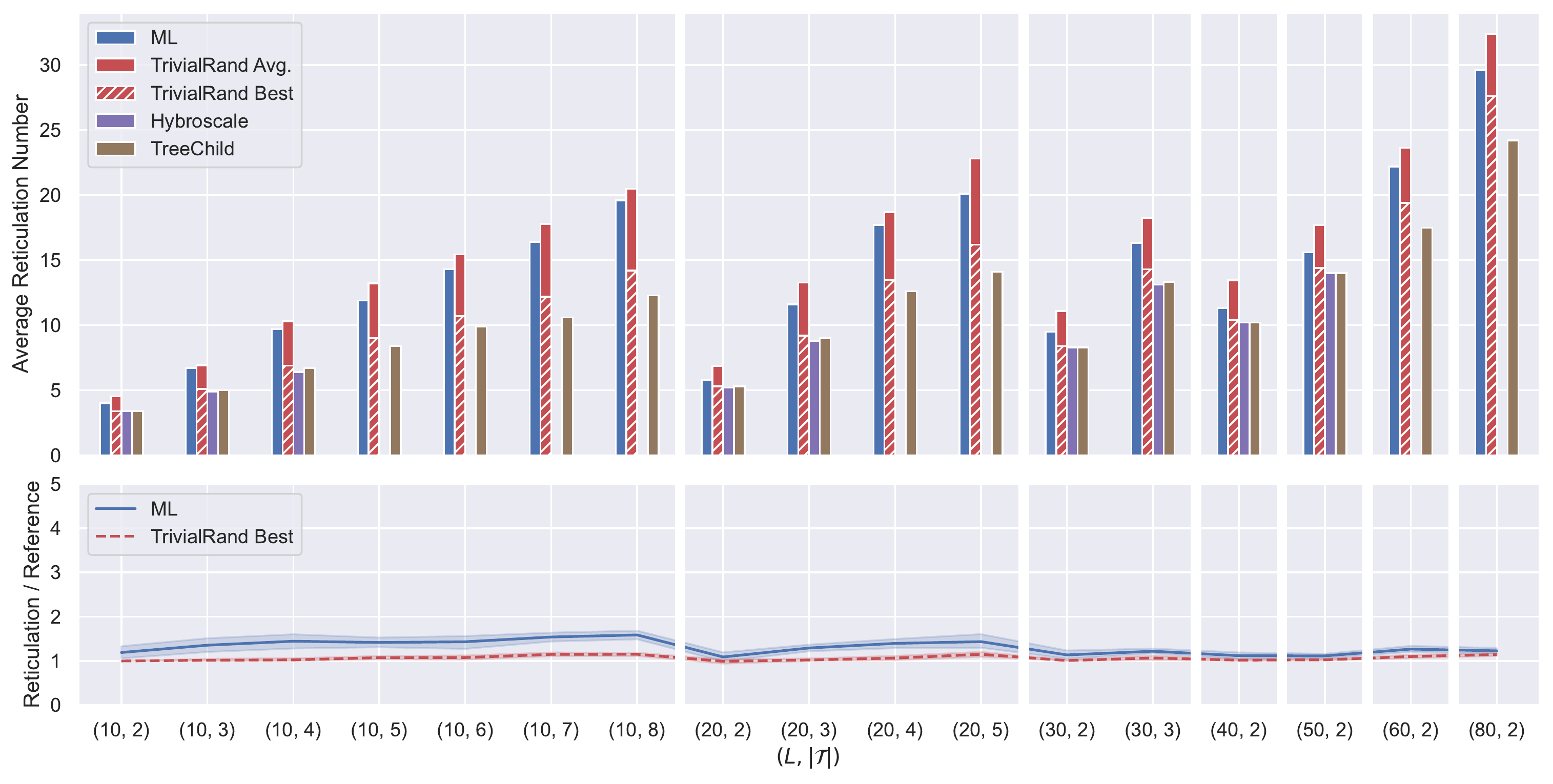}
    \caption{Comparison of \ML{}, \trivialrand{}, \Hybro{}, and \Treechild{} on real data. Each point on the horizontal axis corresponds to one instance group, consisting of 10 instances for a fixed pair $(L,|\cT|)$. In the top graph, the height of each bar gives the average, over all instances of the group, of the number of reticulations returned by the method. The bottom graphs compare the average output of \ML{} within each instance group and the average of the best output given by \trivialrand{} within the group. The shaded areas represent 95\% confidence intervals.} 
    \label{fig:real_instances}
\end{figure}

We tested \ML{} and \trivialrand{} against \Hybro{} and \Treechild{} using the same dataset used in~\cite{van2019practical}, in turn taken from~\cite{beiko2011telling}. The dataset consists of ten instances for each possible combination of the parameters  $L\in\{10,20,30,40,50,60,80,100,150\}$ and $|\cT|\in[2,8]$. In Figure~\ref{fig:real_instances} we show results only for the instance groups for which \Hybro{} or \Treechild{} could output a solution within 1 hour, consistent with the experiments in~\cite{van2019practical}. As a consequence of \Hybro{} and \Treechild{} being exact methods (\Treechild{} only for a restricted class of networks), they performed better than both \ML{} and \trivialrand{} on all instances they could solve, although the best results of \trivialrand{} are often close (no worse than 15\%) and sometimes match the optimal value. 

The main advantage of our heuristics is that they can handle much larger instances than the exact methods.
In the conference version of this paper~\cite{DBLP:conf/wabi/BernardiniIJS22} we showed the results of our heuristics on large real instances, using a ML model trained on 10 networks with at most 100 leaves each. These results demonstrated that consistently with the simulated data, the machine-learned heuristics gave significantly better results than the randomised ones for the largest instances. When we first repeated the experiments with the new models trained on 1000 networks with $\textsf{max}L=100$, however, we did not obtain similar results: instead, the results of the randomised heuristics were better or only marginally worse than the machine-learned ones on almost all the instance groups, including the largest.

Puzzled by these results, we conducted an experiment on the impact of the training set on real data. The results  are reported in Figure~\ref{fig:real_instances_big}, and show that the choice of the networks on which we train our model has a big impact on the quality of the results for the real datasets. This is in contrast with what we observed for the synthetic datasets, for which only the class of the training networks was important, not the specific instances of the networks themselves. According to what was noted in~\cite{van2019practical}, this is most likely due to the fact that the real phylogenetic data have substantially more structure than random synthetic datasets, and the randomly generated training networks do not always reflect this structure. By chance, the networks we used for training the model we used in~\cite{DBLP:conf/wabi/BernardiniIJS22} were similar to real phylogenetic networks, unlike the 1000 networks in the training set of this paper.

\begin{figure}[t]
    \centering
    \includegraphics[width=.85\columnwidth]{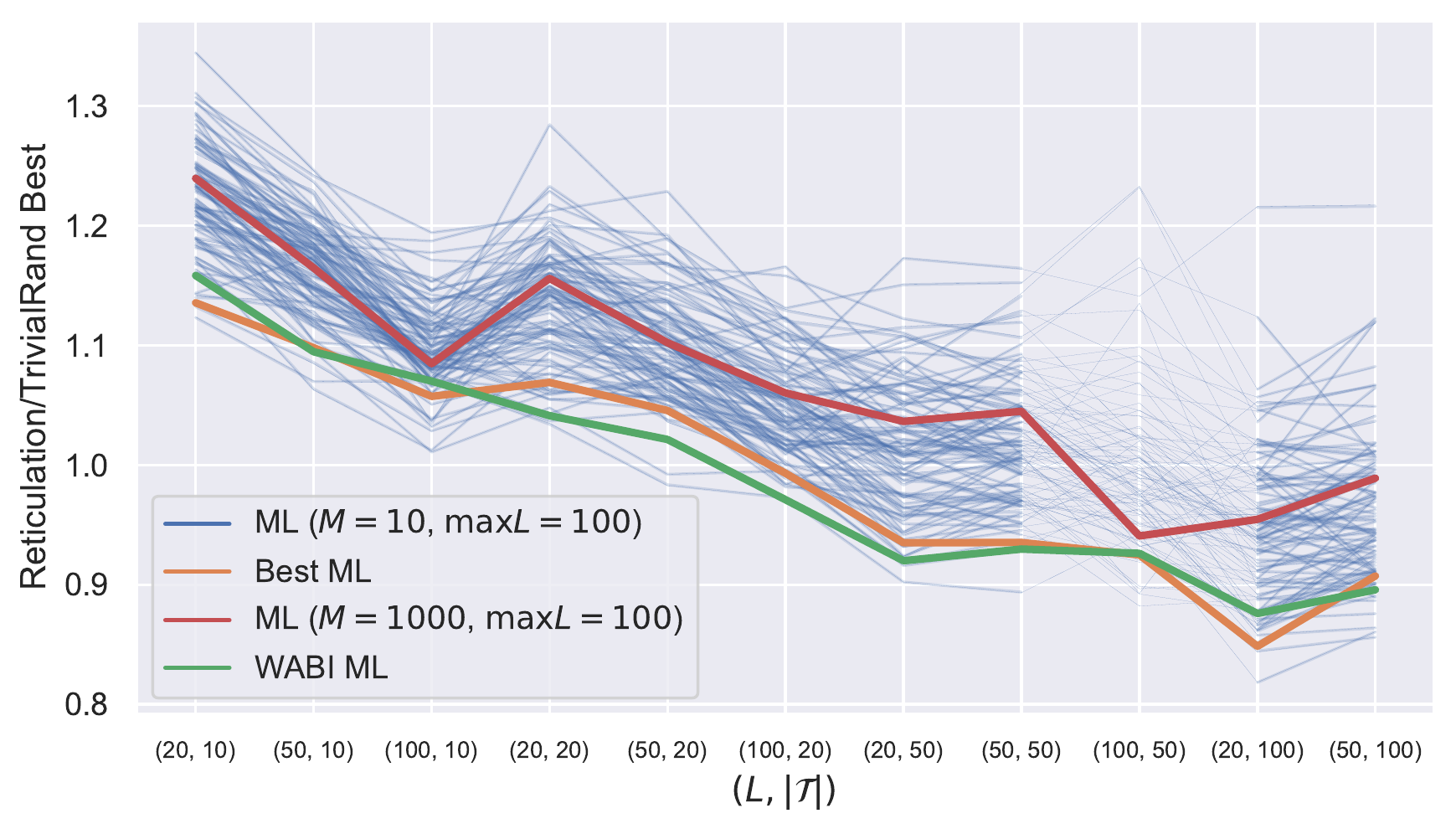}
    \caption{Ratio between the performance of \ML{} and the best value output by \trivialrand{} for different instance groups and different training sets. \trivialrand{} is executed $\min\{x(I),1000\}$ times for each instance $I$, $x(I)$ being the number of runs that could be completed in the same time as one run of \ML{} on $I$. The results are then averaged within each group. Each blue line represents the results obtained training the model with a different set of 10 randomly generated LGT networks with at most 100 leaves each. The green line corresponds to the training set used in~\cite{DBLP:conf/wabi/BernardiniIJS22}; the orange line represents one of the best-performing sets; the red line corresponds to the training set we used for the experiments on LGT and ZODS data in this paper, consisting of 1000 randomly generated LGT networks.}
    \label{fig:real_instances_big}
\end{figure}

\subsubsection{Experiments on Scalability}\label{exp:scalability} We conducted experiments to study how the running time of our heuristics scales with increasing instance size for all datasets. In Figure~\ref{fig:runtimes} we report the average of the running times of \ML{} for the instances within each instance group with a 95\% confidence interval, for an increasing number of reticulations (synthetic datasets) or number of trees (real dataset). The datasets and the instance groups are those described in the previous sections. 
Note that we did not report the running times of the randomised heuristics because they are meant to be executed multiple times on each instance, and in all the experiments we bounded the number of executions precisely using the time required for one run of \ML{}.

We also compared the running time of our heuristics with the running times of the exact methods \Treechild{} and \Hybro{}. The results are shown in Figure~\ref{fig:real_runtime} and are consistent with the execution times of the exact methods growing exponentially, while the running time of our heuristics grows polynomially. Note that networks with more reticulations are reduced by longer CPS and thus the running time increases with the number of reticulations.

\begin{figure}[t]
    \centering
        \subfloat[Normal]{\includegraphics[height=4cm]{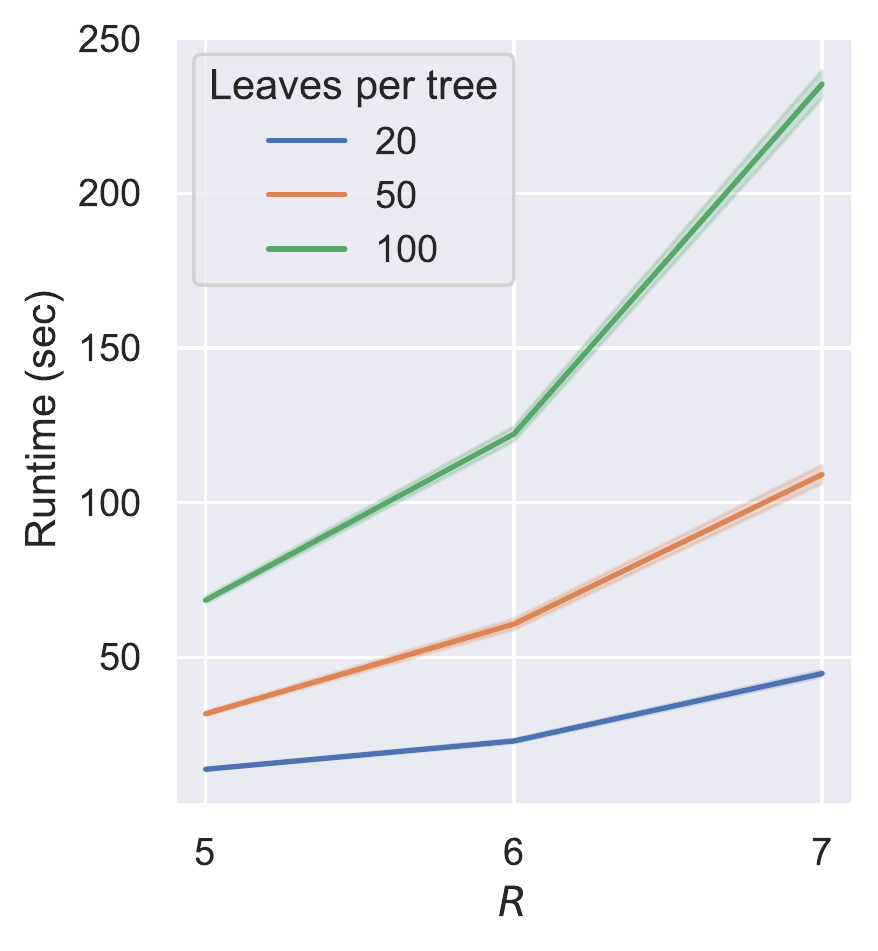}
        }
        \subfloat[LGT]
        {\includegraphics[height=4cm]{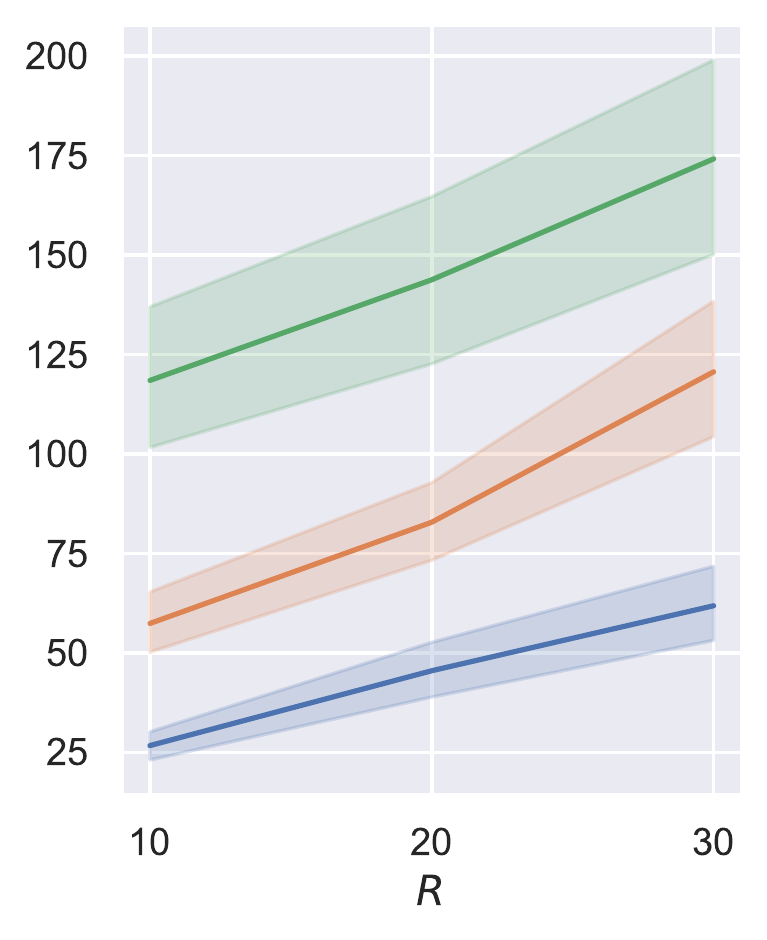}
        }
        \subfloat[ZODS]
        {\includegraphics[height=4cm]{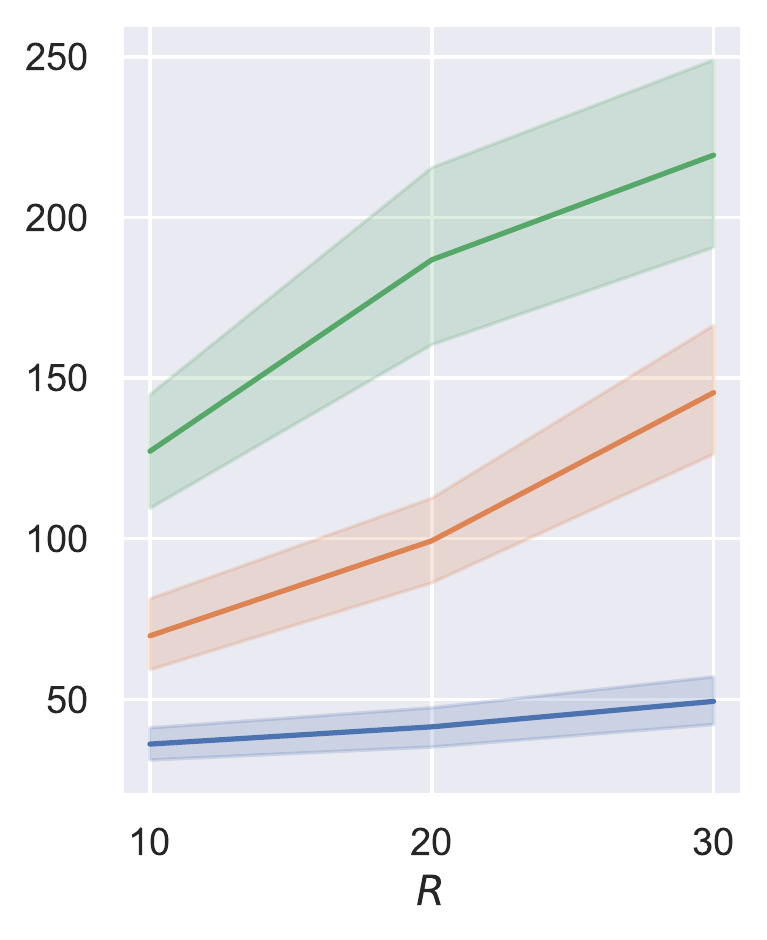}
        }
        \subfloat[Real]
        {\includegraphics[height=4cm]{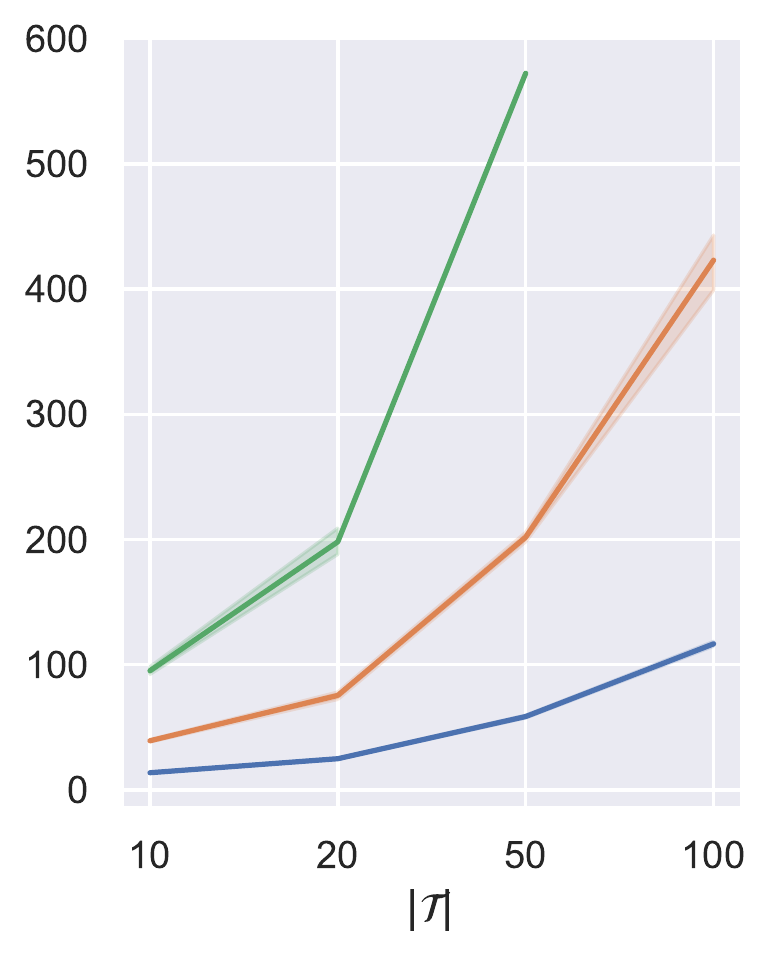}
        }
    \caption{The running time (in seconds) of \ML{} for the instance groups described in Sections~\ref{exp:normal}, \ref{exp:LGT}, \ref{exp:ZODS}, \ref{exp:real}. The solid lines represent the average of the running times for the instances within each instance group. The shaded areas represent 95\% confidence intervals.} 
    \label{fig:runtimes}
\end{figure}

\begin{figure}[t]
    \centering
    \includegraphics[width=\columnwidth]{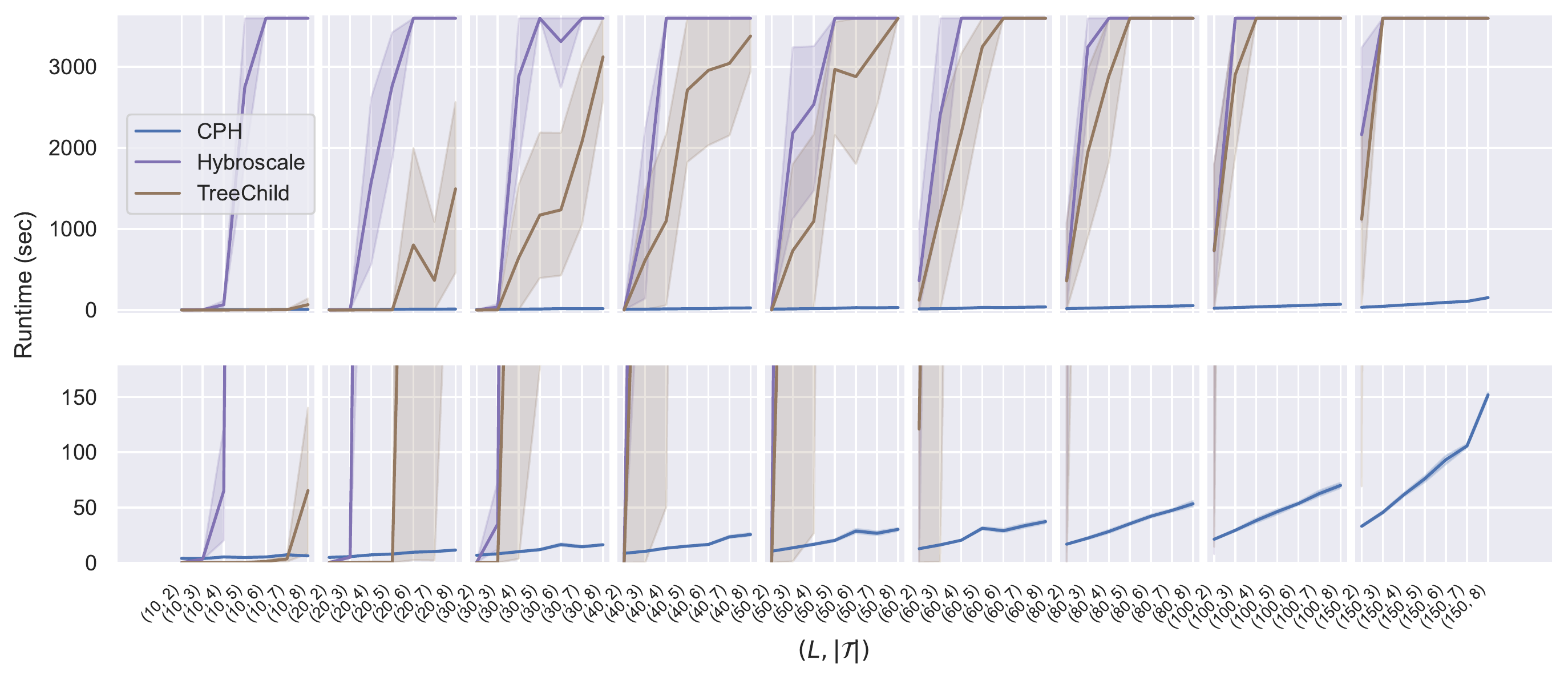}
    \caption{The running time of \ML{} on the real dataset described in Section~\ref{exp:real} compared with the running time of the exact methods \Hybro{} and \Treechild{} on the same dataset. The solid lines represent the average running times within each instance group. The shaded areas represent 95\% confidence intervals.}
    \label{fig:real_runtime}
\end{figure}

\subsubsection{Experiments on Non-Exhaustive Input Trees}\label{exp:missing_leaves} The instances on which we tested our methods so far all consisted of a set of exhaustive trees, that is, each input tree had the same set of leaves which coincided with the set of leaves of the network. However, this is not a requirement of our heuristics, which are able to produce feasible solutions also when the leaf sets of the input trees are different, that is when their leaves are proper subsets of the leaves of the optimal networks that display them. 

To test their performance on this kind of data, we generated 18 LGT instance groups starting from the instances we used in Section~\ref{exp:LGT} and removing a certain percentage $p$ of leaves from each tree in each instance uniformly at random. Specifically, we generated an instance group for each value of $p\in\{5,10,15,20,25,50\}$ starting from the LGT instance groups with $L=100$ leaves and $R\in\{10,20,30\}$ reticulations. Since the performances of the two machine-learned heuristics were essentially overlapping for all of the other experiments, and since \trivialrand{} performed consistently better than the other randomised heuristics, we limited this test to \ML{} and \trivialrand{}. The results are shown in Figure~\ref{fig:missing_leaves}.

In accordance with intuition, the performance of both methods decreases with an increasing percentage of removed leaves, as the trees become progressively less informative. However, the degradation in the quality of the solutions is faster for \ML{} than for \trivialrand{}, consistent with the fact that \ML{} was trained on exhaustive trees only: when the difference between the training data and the input data becomes too large, the behaviour of the machine-learned heuristic becomes unpredictable.
We demand the design of algorithms better suited for trees with missing leaves for future work.

\begin{figure}[t]
    \centering
    \includegraphics[width=.85\columnwidth]{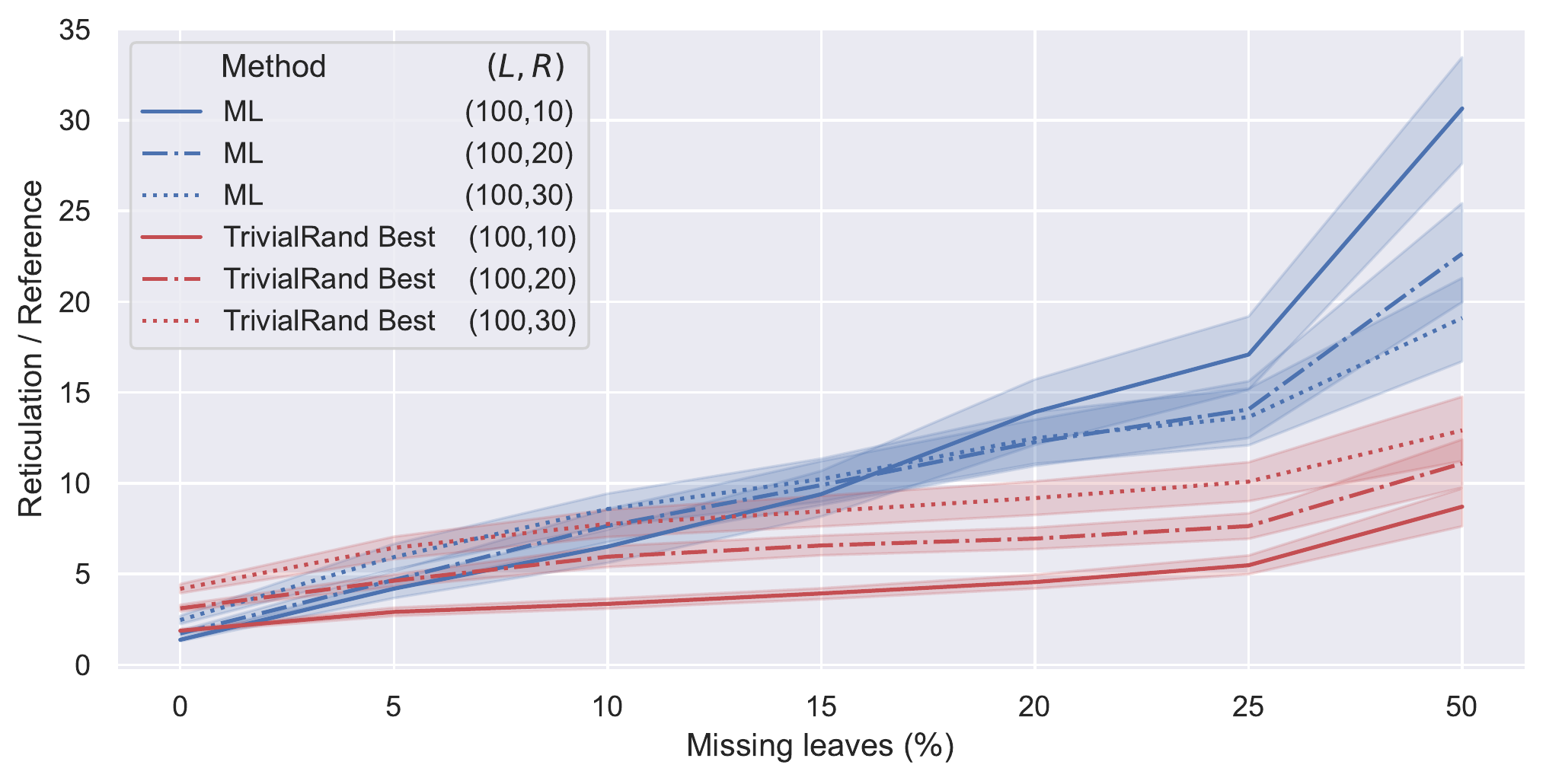}
    \caption{Ratio between the number of reticulations outputted by \ML{} and \trivialrand{} Best and the reference value for an increasing percentage of removed leaves on LGT data. Each point on the horizontal axis corresponds to a certain percentage of leaves removed from each tree; each line represents the average, within the instances of a group $(L,R)$ with a certain percentage of removed leaves, of the output reticulation number divided by the reference value. The shaded areas represent 95\% confidence intervals.}
    \label{fig:missing_leaves}
\end{figure}

\subsubsection{Effect of the Threshold on ML.}\label{exp:threshold}
We tested the effectiveness of adding a threshold $\tau>0$ to \ML{} on the same datasets of Sections~\ref{exp:normal}, \ref{exp:LGT} and~\ref{exp:ZODS} (normal, LGT and ZODS). Recall that each instance group consists of 48 instances.
We ran \ML{} ten times for each threshold
$\tau\in\{0,0.1,0.3,0.5,0.7\}$ on each instance, took the lowest output reticulation number and averaged these results within each instance group. 

The results are shown in Figure~\ref{fig:threshold}. For all types of data, a threshold $\tau\leq 0.3$ is beneficial, intuitively indicating that when the probability of a pair being reducible is small it gives no meaningful indication, and thus random choices among these pairs are more suited. The seemingly best value for the threshold, though, is different for different types of instances. 
The normal instances seem to benefit from quite high values of $\tau$, the best among the tested values being $\tau=0.7$. While the optimal $\tau$ value for normal instances could be even higher, we know from Figure~\ref{fig:results_normal} that it must be $\tau<1$, as the random strategies are less effective than the one based on machine learning for normal data. For the LGT and the ZODS instances, the best threshold seems to be around $\tau=0.3$, while very high values ($\tau=0.7$) are counterproductive. This is especially true for the LGT instances, consistent with the randomised heuristics being less effective for them than for the other types of data (see Figure~\ref{fig:results_LGT}).

These experiments should be seen as an indication that introducing some randomness may improve the performance of the ML heuristics, at the price of running them multiple times. 
We defer a more thorough analysis to future work.
\begin{figure}[h]
    \centering
        \subfloat[Normal]
        {
        \includegraphics[width=0.32\columnwidth]{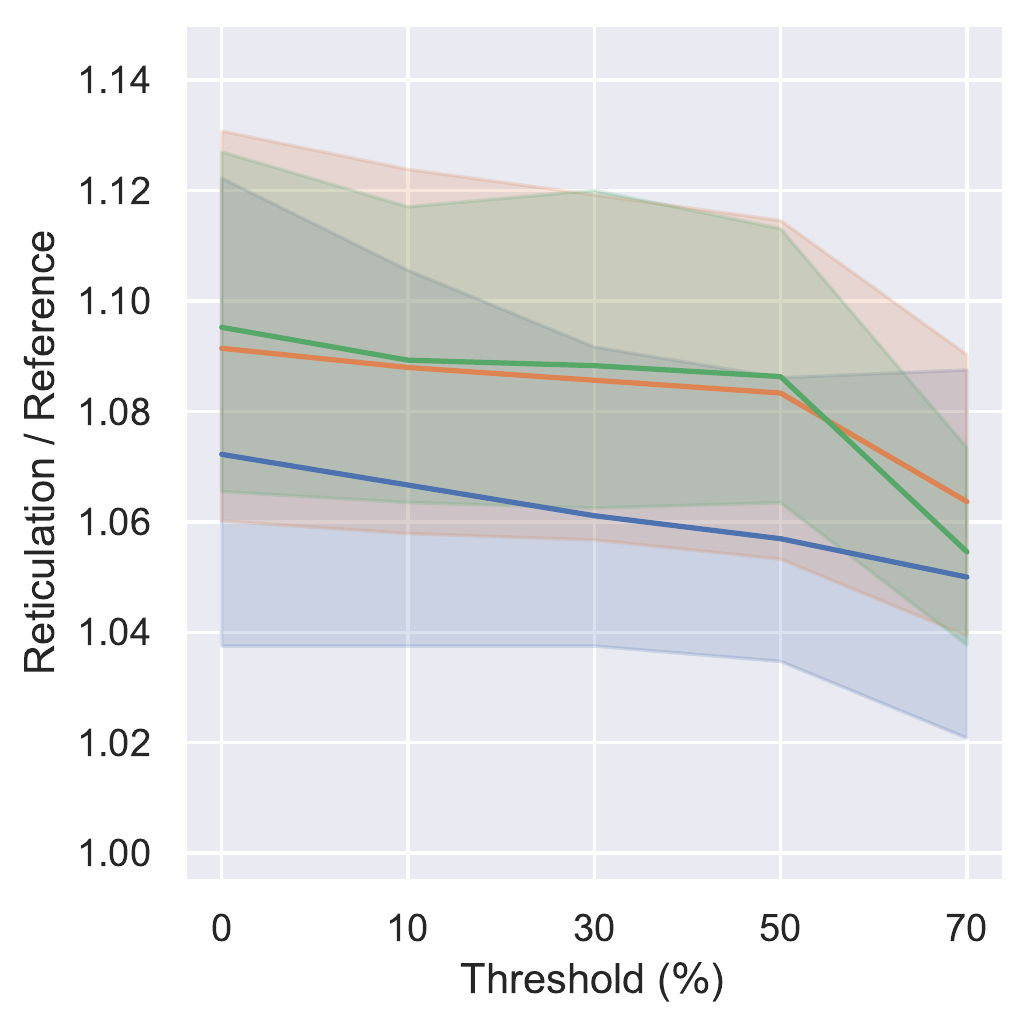}
        }
        \subfloat[LGT]
        {
         \includegraphics[width=0.32\columnwidth]{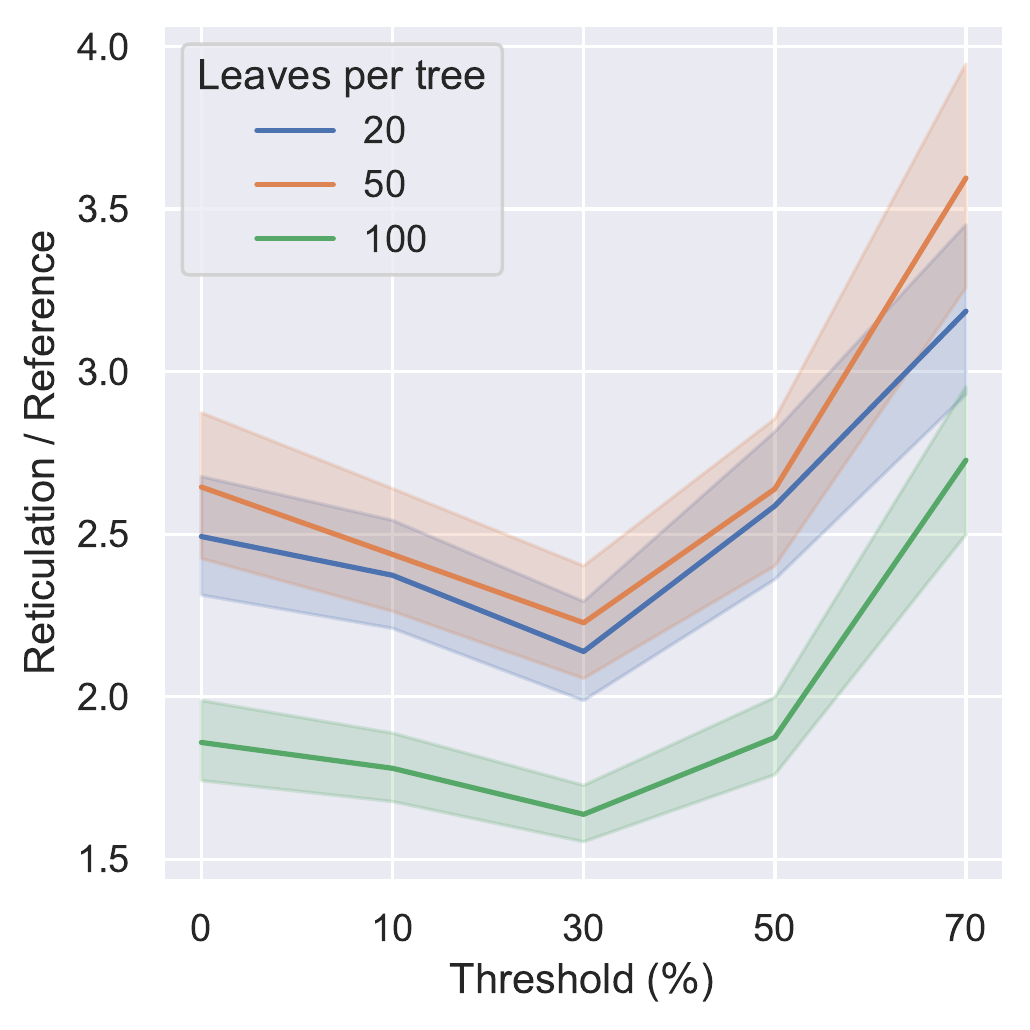}
        }
        \subfloat[ZODS]
        {
         \includegraphics[width=0.32\columnwidth]{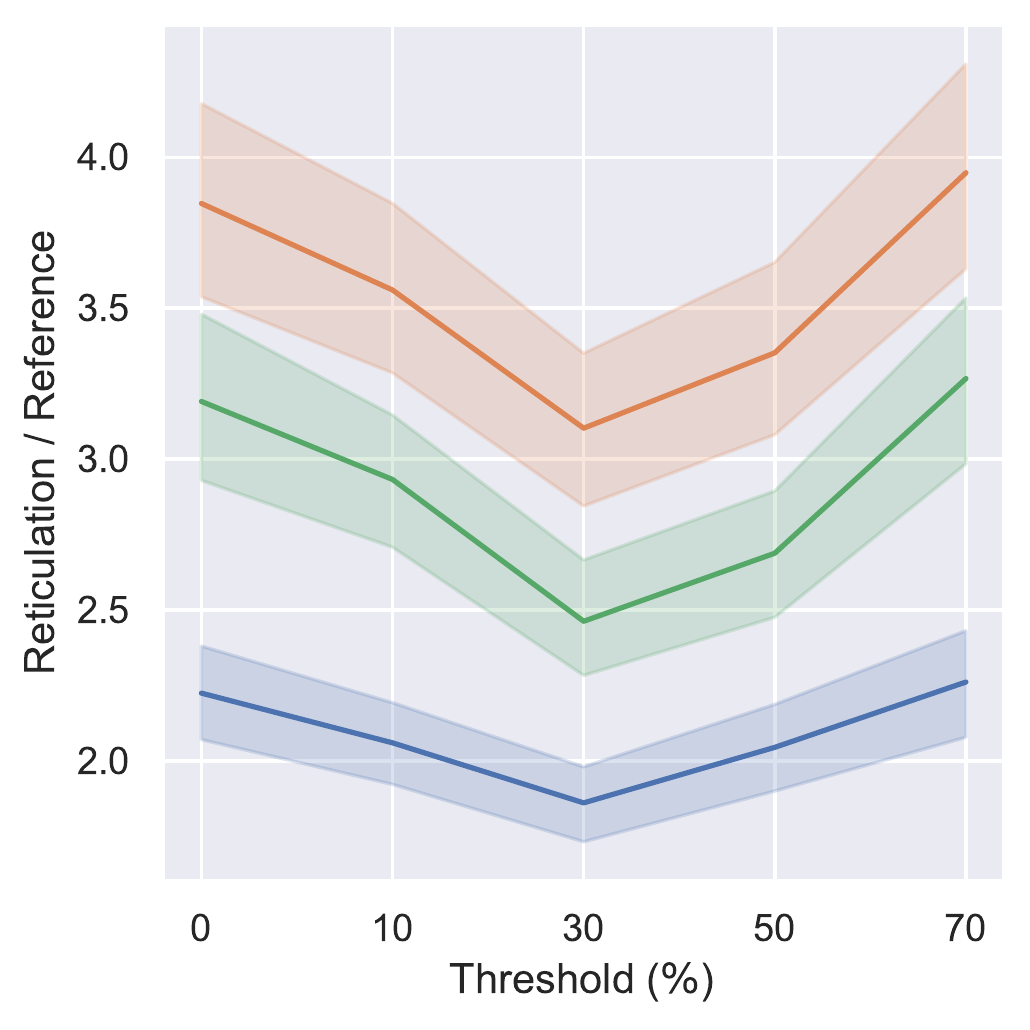}
        }
    \caption{The reticulation number when running \ML{} with different thresholds on the instance groups of Sections~\ref{exp:normal}, \ref{exp:LGT} and~\ref{exp:ZODS}. Each instance was run 10 times, and the lowest reticulation value of these runs was selected. The shaded areas represent 95\% confidence intervals.}
    \label{fig:threshold}
\end{figure}
\subsection{A Non-Learned Heuristic Based on Important Features}\label{sec:new_heuristic}
In this section we propose \textsf{FeatImp}, yet another heuristic in the CPH framework. Although \textsf{FeatImp} does not rely on a machine learning model, we defined the rules to choose a cherry on the basis of the features that were found to be the most relevant according to the model we used for \ML{} and \trivML{}.

To identify the most suitable rules, we trained a classification tree using the same features and training data as the ones used for the \textsf{ML} heuristic (see figure~\ref{fig:classificationTree} in Appendix~\ref{app:random_forests}). We then selected the most relevant features used in such tree and used them to define the function \textsf{PickNext} listed by Algorithm~\ref{alg:FeatImp}: namely, the features $4$, $8_t$, $11_d$ and $12_t$ of Table~\ref{tab:cherry_features} (the ratio of trees having both leaves $x$ and $y$ in which $(x,y)$ is reducible, the average of the topological leaf distance between $x$ and $y$ scaled by the depth of the trees, the average of the ratios $d(x,\textsf{LCA}(x,y))/d(y,\textsf{LCA}(x,y))$ and the average of the topological distance from $x$ to the root over the topological distance from $y$ to the root, respectively).

To compute and update these quantities we proceed as described in Section~\ref{sub:time_complexity} and Appendix~\ref{app:time_complex}. The general idea of the function  \textsf{PickNext} used in \textsf{FeatImp} is to mimic the first splits of the classification tree by progressively discarding the candidate reducible pairs that are not among the top $\alpha\%$ scoring for each of the considered features, for some input parameter $\alpha$.

\begin{algorithm}[H]
\small
   \caption{Function \textsf{PickNext} used in \textsf{FeatImp}}  \label{alg:FeatImp}
\begin{algorithmic}[1]
\Statex \textbf{INPUT:} A set $\cT$ of phylogenetic trees and a parameter $\alpha\in(0,100)$.
 \Statex \textbf{OUTPUT:} Next cherry to pick $(x, y)$ 
    \If{there exists a trivial cherry}
         \State Select a trivial cherry $(x, y)$ uniformly at random;
    \Else
        \State $C \gets$ all reducible pairs of $\cT$;
        \State $C \gets$ the $\alpha \%$ cherries of $C$ with the highest value for feature $4$;
        \State $C \gets$ the $\alpha \%$ cherries of $C$ with the highest value for feature $8_t$;
        \State $C \gets$ the $\alpha \%$ cherries of $C$ with the highest value for feature $11_d$;
        \State\label{lastline} $(x, y)\gets$ the pair of $C$ with the highest value for feature $12_t$;
    \EndIf
    \State \Return $(x, y)$;
  \end{algorithmic}
\end{algorithm}

We implemented \textsf{FeatImp} and test it on the same instances as sections~\ref{exp:normal}, \ref{exp:LGT} and \ref{exp:ZODS} with $\alpha=20$. The results are shown in Figure~\ref{fig:feat_imp_heur}. 
As expected, \textsf{FeatImp} works consistently worse than \ML{} on all the tested datasets, and it also performs worse than \trivialrand{} on most instance groups. However, it is on average 12\% better than \trivialrand{} on the LGT instance group having 50 leaves and 30 reticulations and on all the LGT instance groups with 100 leaves, which are the most difficult for the randomised heuristics, as already noticed in Section~\ref{exp:LGT}. The results it provides for such difficult instances are only on average 20\% worse than those of \ML, with the advantage of not having to train a model to apply the heuristic.

These experiments are not intended to be exhaustive, but should rather be seen as an indication that machine learning can be used as a guide to design smarter non-learned heuristics. Possible improvements of \textsf{FeatImp} include using different values of $\alpha$ for different features, introducing some randomness in Line~\ref{lastline}, that is, instead of choosing the single top scoring pair to choose one among the top $\alpha\%$ at random, or to use fewer/more features.

\begin{figure}[t]
    \centering
        \subfloat[Normal]
        {
         \includegraphics[width=0.32\columnwidth]{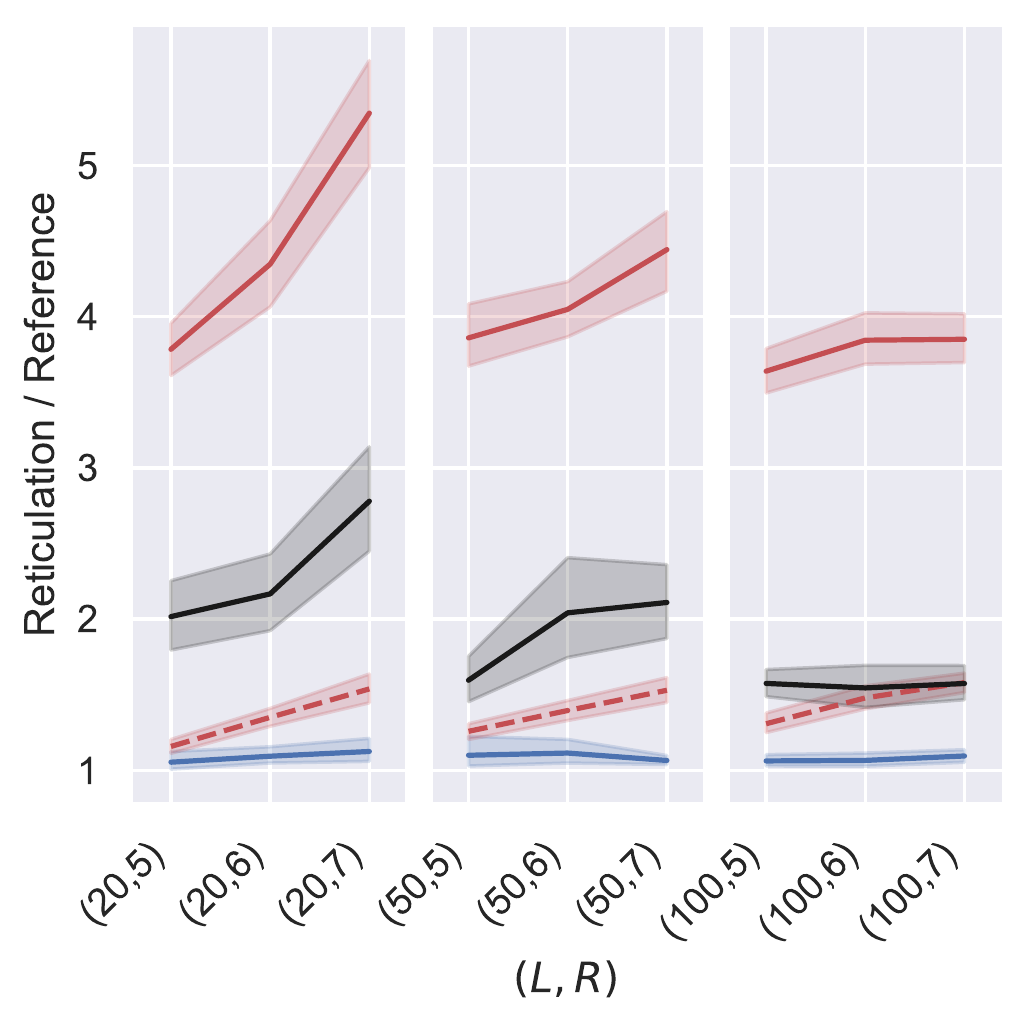}
        }
        \subfloat[LGT]
        {
        \includegraphics[width=0.32\columnwidth]{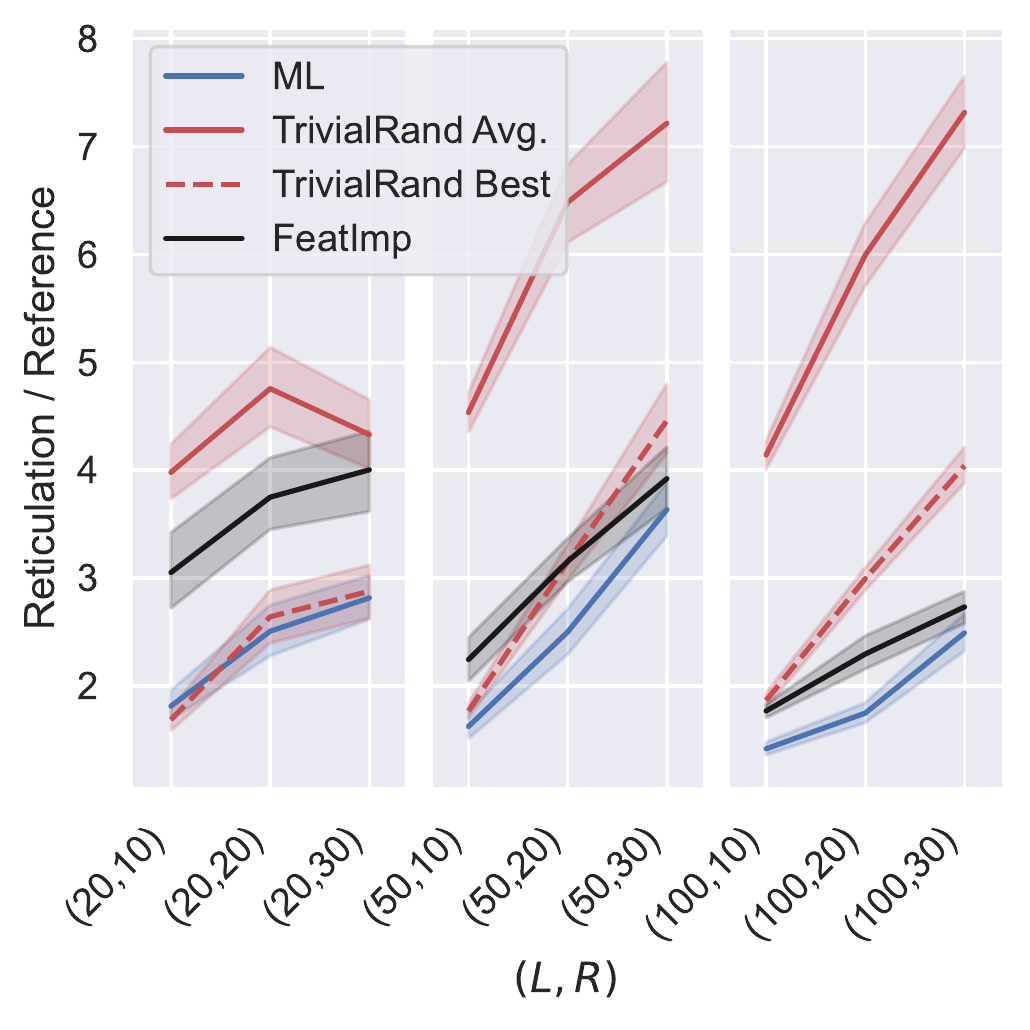}
        }
        \subfloat[ZODS]{
        \includegraphics[width=0.32\columnwidth]{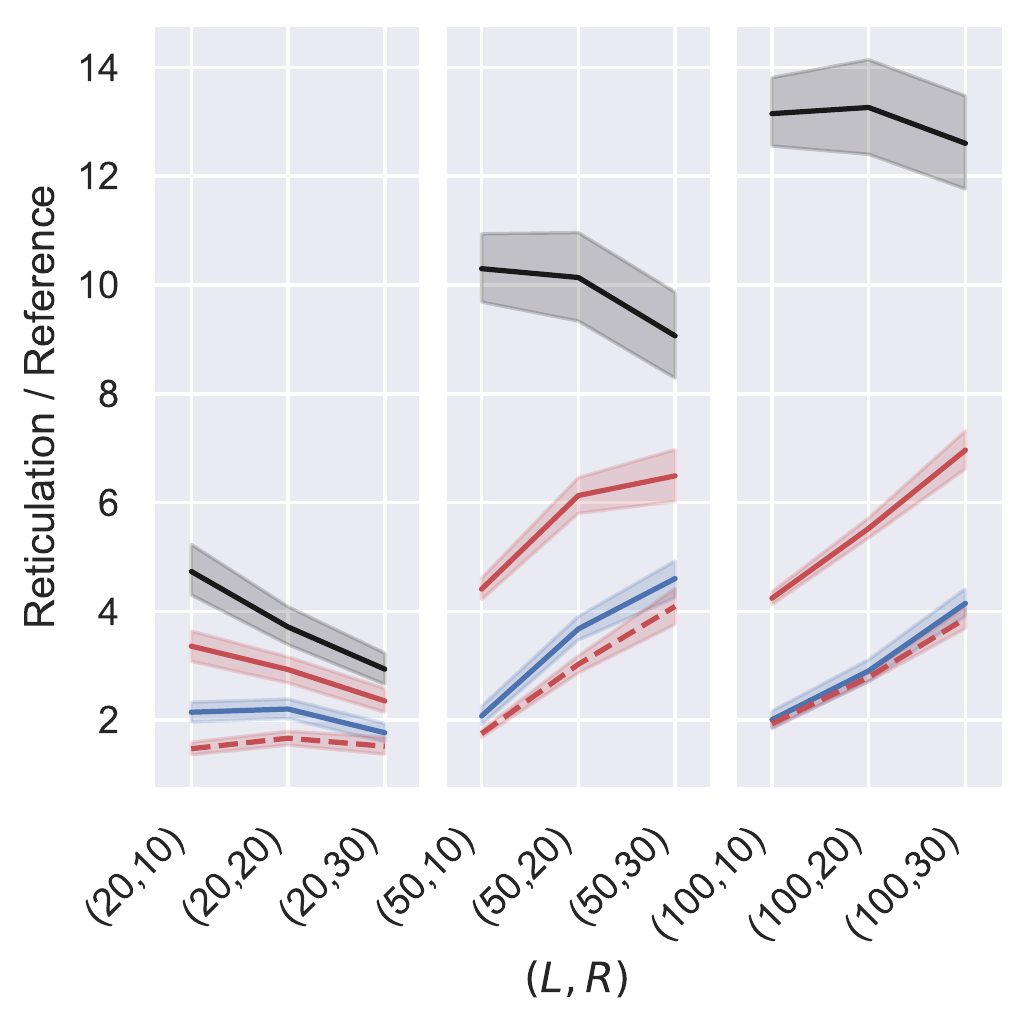}
        }
    \caption{Comparison of the results of \featimp{}, \ML{} and \trivialrand{} on the instance groups described in Sections~\ref{exp:normal}, \ref{exp:LGT} and~\ref{exp:real}.
    Each point on the horizontal axis corresponds to an instance group; each line represents the average, within the instance group, of the output reticulation number divided by the reference value. The shaded areas represent 95\% confidence intervals.}
    \label{fig:feat_imp_heur}
\end{figure}

\section{Conclusions}\label{sec:discussion}
Our contributions are twofold: first, we presented the first methods that allow reconstructing a phylogenetic network from a large set of large binary phylogenetic trees. Second, we show the promise and the limitation of the use of machine learning in this context.
Our experimental studies indicate that machine-learned strategies, consistent with intuition, are very effective when the training data have a structure similar enough to the test data. In this case, the results we obtained with machine learning were the best among all the tested methods, and the advantage is particularly evident in the most difficult instances. Furthermore, preliminary experiments indicate that the performance of the machine-learned methods can even be improved by introducing appropriate thresholds, in fact mediating between random choices and predictions. However, when the training data do not sufficiently reflect the structure of the test data, repeated runs of the fast randomised heuristics lead to better results.
The non-learned cherry-picking heuristic we designed based on the most relevant features of the input (identified using machine learning) shows yet another interesting direction.

 Our results suggest many interesting directions for future work. First of all, we have seen that machine learning is an extremely promising tool for this problem since it can identify cherries and reticulated cherries of a network, from displayed trees, with very high accuracy. It would be interesting to prove a relationship between the machine-learned models' accuracy and the produced networks' quality. In addition, do there exist algorithms that exploit the high accuracy of the machine-learned models even better? Could other machine learning methods than random forests, or more training data, lead to even better results? Our methods are applicable to trees with missing leaves but perform well only if the percentage of missing leaves is small. Can modified sets of features be defined that are more suitable for input trees with many missing leaves? Moreover, we have seen that combining randomness with machine learning can lead to better results than either individual approach. However, we considered only one strategy to achieve this. What are the best strategies for combining randomness with machine learning for this, and other, problems? From a practical point of view, it is important to investigate whether our methods can be extended to deal with nonbinary input trees and to develop efficient implementations: in fact, we point out that our current implementations are in Python and not optimised for speed. Faster implementations could make machine-learned heuristics with nonzero thresholds even more effective. Finally, can the machine-learning-based approach be adapted to other problems in the phylogenetic networks research field?

 \appendix
\section{Time Complexity}\label{app:time_complex}

\lemfeaturesA*
\begin{proof}
Let $F_{(x,y)}^i$ denote the current value of the $i$-th feature for a cherry $(x,y)$. When reducing a cherry $(x,y)$ in a tree $T$ (thus deleting $x$ and $p(x)=p(y)$ and then adding a direct edge from $p(p(y))$ to $y$), we check whether the other child of $p(p(y))$ is a leaf $z$ or not. If not, no new cherry is created in $T$, thus the features 1-4 remain unaffected for all the cherries of $\cT$. 
Otherwise, $(z,y)$ and $(y,z)$ are new cherries of $T$ and we can distinguish two cases.
\begin{enumerate}
    \item $(z,y)$ and $(y,z)$ are already cherries of $\cT$. Then, $F^1_{(y,z)}$ and $F^1_{(z,y)}$ are increased by $\frac{1}{|\cT|}$; $F^4_{(y,z)}$ and $F^4_{(z,y)}$ are increased by $\frac{1}{|\cT^{y,z}|}$, where  $|\cT^{y,z}|$ is the number of trees that contain both $y$ and $z$ and is equal to $|\cT|F^5_{(y,z)}$.
    To update features 2 and 3 we use two auxiliary data structures $\textsf{new\_cherries}_{(y,z)}$ and $\textsf{new\_cherries}_{(z,y)}$ to collect the distinct cherries that would originate after picking $(y,z)$ and $(z,y)$ in each tree, respectively. These structures must allow efficient insertions, membership queries, and iteration over the elements\footnote{For example, hashtables paired with lists.}, and can be deleted before picking the next cherry in $\cT$. 
    If the other child of $p(p(z))$ is a leaf $w$, we add $(z,w)$ and $(w,z)$ to $\textsf{new\_cherries}_{(y,z)}$ and $(y,w)$ and $(w,y)$ to $\textsf{new\_cherries}_{(z,y)}$ (unless they are already present). 
    \item $(z,y)$ and $(y,z)$ are new cherries of $\cT$. Then we insert them into $\cherries$. We initially set $F^1_{(y,z)}=F^1_{(z,y)}=\frac{1}{|\cT|}$, and for features 2-3 we create the same data structures as the previous case. To compute $F^5_{(y,z)}=F^5_{(z,y)}$ we first compute $|\cT^{y,z}|$ by checking whether $y$ and $z$ are both leaves of $T$ for each $T\in\cT$. Then we set $F^5_{(y,z)}=F^5_{(z,y)}=\frac{|\cT^{y,z}|}{|\cT|}$ and  $F^4_{(y,z)}=F^4_{(z,y)}=\frac{1}{|\cT^{y,z}|}$.%
\end{enumerate}
 Once we have reduced $(x,y)$ in all trees, we count the elements of each of the auxiliary data structures $\textsf{new\_cherries}$ and update features 2-3 of the corresponding cherries accordingly. Since picking a cherry can create up to two new cherries in each tree, and for each new cherry we add up to two elements to an auxiliary data structure, 
 this step requires $\cO(|\cT|)$ time for each iteration.

Feature 5 must be updated for all the cherries corresponding to the unordered pairs $\{x,w\}$ with $w\neq y$. 
To do so, when we reduce $(x,y)$ in a tree $T$ we go over its leaves: for each leaf $w\neq y$ we decrease $F^5_{(x,w)}$ and $F^5_{(w,x)}$ by $\frac{1}{|\cT|}$ (if $(x,w)$ and $(w,x)$ are currently cherries of $\cT$).
This requires $\cO(|X|^2)$ total time per tree over all the iterations, because we scan the leaves of a tree only when we reduce a cherry in that tree.
Computing feature 5 when new cherries of $\cT$ are created (case 2) requires constant time per tree per cherry. The total number of cherries created in $\cT$ over all the iterations cannot exceed $2||\cT||$, thus the total time required to update feature 5 is $\cO(|\cT|(||\cT||+|X|^2))$. We arrived at the following result.
\end{proof}
 
\lemfeaturesB*
\begin{proof}
Recall that during the initialization phase, we store the depth of each node, both topological and with respect to the branch lengths, and we preprocess each tree to allow constant-time LCA queries. Note that reducing cherries in the trees does not affect the height of the nodes nor their ancestry relations, thus it suffices to preprocess the tree set only once at the beginning of the algorithm.

When we reduce a cherry $(x,y)$ in a tree $T$, this may affect the depth of $T$ as a consequence of the internal node $p(x)$ being deleted. We thus visit $T$ to update its depth (both topological and with the branch lengths), and after updating the depth of all trees, we update the maximum value over the whole set $\cT$ accordingly. In order to describe how to update the features $6_{d,t}-12_{d,t}$ we denote by $\textsf{old\_depth}^t(T)$ the topological depth of $T$ before reducing $(x,y)$, $\textsf{new\_depth}^t(T)$ its depth after reducing $(x,y)$, and use analogous notation for the distances $\textsf{old\_dist}^t$ and $\textsf{new\_dist}^t$ between two nodes of a tree and for the depth, the max depth, and distances with the branch lengths.

Whenever the value of the maximum topological depth changes, we update the value of feature $6_t$ for all the current cherries $(z,w)$ as $F^{6_{t}}_{(z,w)}=\frac{F^{6_{t}}_{(z,w)}\cdot\textsf{old\_max\_depth}^t}{\textsf{new\_max\_depth}^t}$. Since the maximum topological depth can change $\cO(|X|)$ times over all the iterations, and the total number of cherries at any moment is $\cO(|\cT||X|)$, these updates require $\cO(|\cT||X|^2)$ total time. We do the same for feature $6_d$, but since the maximum branch-length depth can change once per iteration in the worst case, this requires $\cO(||\cT||^2)$ time overall. 

Features $8_{d,t}-12_{d,t}$ must be then updated to remove the  contribution of $T$ for the cherries $(x,w)$ and $(w,x)$ for each leaf $w\neq x\neq y$ of $T$, because $x$ and $w$ will no longer appear together in $T$. 
These updates require $\cO(1)$ time per leaf and can be done as follows. 
We set 
\begin{equation}\label{eq:avg_update}
F^{8_{t}}_{(x,w)}=\frac{F^{8_{t}}_{(x,w)}\cdot|\cT^{x,w}|-\frac{\textsf{old\_dist}^{t}(x,w)}{\textsf{old\_depth}^t(T)}}{|\cT^{x,w}|-1}\end{equation}
and use analogous formulas to update $F^{8_{d}}_{(x,w)}$ and features $9_{d,t}-12_{d,t}$ for $(x,w)$ and $(w,x)$.

We finally need to further update all the features $6_{d,t}-12_{d,t}$ for all the cherries of a tree $T$ in which $(x,y)$ has been reduced and whose depth has changed, including the newly created ones. This can be done in $\cO(1)$ time per cherry per tree with opportune formulas of the form of Equation~\ref{eq:avg_update}. We have obtained the stated bound.
\end{proof}

\section{Random Forest Models}\label{app:random_forests}
\renewcommand{\arraystretch}{0.75}

\begin{table}[h]
\scriptsize
\caption{Feature importances of random forest trained on the biggest dataset ($M=1000$ and $\max L=100$) based on normal (a) and LGT (b) network data. Higher importance indicates that a feature has more effect on the trained model. The values sum up to one. The descriptions of the features are given in Table \ref{tab:cherry_features}.}\label{tab:feature_importances}
\begin{subtable}{.5\linewidth}
  \centering
    \caption{Normal}
\begin{tabular}{lcr}
\toprule
      Features &  &  Importance \\
\midrule
 Leaf distance & ($t$) &       0.190 \\
       Trivial &     &       0.155 \\
Cherry in tree &     &       0.143 \\
 Leaf distance & ($d$) &       0.122 \\
  LCA distance & ($t$) &       0.068 \\
     Depth $x/y$ & ($t$) &       0.050 \\
  Cherry depth & ($t$) &       0.047 \\
     Depth $x/y$ & ($d$) &       0.043 \\
  LCA distance & ($d$) &       0.028 \\
  Leaf depth $x$ & ($t$) &       0.023 \\
  Leaf depth $y$ & ($t$) &       0.023 \\
  Cherry depth & ($d$) &       0.020 \\
  Leaf depth $x$ & ($d$) &       0.020 \\
  Leaf depth $y$ & ($d$) &       0.020 \\
  Before/after &     &       0.015 \\
    Tree depth & ($d$) &       0.012 \\
    Tree depth & ($t$) &       0.011 \\
  New cherries &     &       0.006 \\
Leaves in tree &     &       0.004 \\
\bottomrule
\end{tabular}
    \end{subtable}%
    \begin{subtable}{.5\linewidth}
      \centering
        \caption{LGT}
\begin{tabular}{lcr}
\toprule
      Features &  &  Importance \\
\midrule
       Trivial &     &       0.184 \\
 Leaf distance & ($t$) &       0.162 \\
Cherry in tree &    &       0.146 \\
 Leaf distance & ($d$) &       0.114 \\
     Depth $x/y$ & ($t$) &       0.058 \\
  LCA distance & ($t$) &       0.056 \\
  Cherry depth & ($t$) &       0.045 \\
     Depth $x/y$ & ($d$) &       0.038 \\
  LCA distance & ($d$) &       0.032 \\
  Leaf depth $y$ & ($t$) &       0.024 \\
  Leaf depth $x$ & ($t$) &       0.023 \\
  Cherry depth & ($d$) &       0.023 \\
  Leaf depth $y$ & ($d$) &       0.022 \\
  Leaf depth $x$ & ($d$) &       0.022 \\
  Before/after &    &       0.016 \\
    Tree depth & ($d$) &       0.013 \\
    Tree depth & ($t$) &       0.011 \\
  New cherries &     &       0.006 \\
Leaves in tree &     &       0.003 \\
\bottomrule
\end{tabular}
    \end{subtable} 
\end{table}

\begin{table}
\scriptsize
\caption{Trained random forest models on different datasets for different combinations of $\max L$ (maximum number of leaves per network) and $M$ (number of networks). Each row in the table represents one model. For each model, the testing accuracy is given under ``Accuracy'', and the total number of data points retrieved from all $M$ networks is given under ``Num. data''. Each dataset is split for training and testing ($90\%-10\%$). The training duration for the random forest is given in column ``Training'' and the time needed to generate the training data is given in column ``Data gen.'', in hours per core (we used 16 cores in total).}\label{tab:ml_description}
\begin{subtable}{.8\linewidth}
  \centering
    \caption{Normal}
\begin{tabular}{lrcrcc}
\toprule
 $\max L$   &  $M$    & Accuracy &  Num. data & Training (min) & Data gen. (hour/core) \\
 \midrule
20  & 5    &      1.0 &        840 &               00:00 &                       00:00:12 \\
    & 10   &    0.994 &       1,804 &               00:00 &                       00:00:22 \\
    & 100  &    0.998 &      17,388 &               00:03 &                       00:04:19 \\
    & 500  &    0.994 &      73,168 &               00:16 &                       00:15:18 \\
    & 1000 &    0.993 &     151,308 &               00:42 &                       00:29:49 \\
50  & 5    &    0.994 &       3,580 &               00:00 &                       00:01:21 \\
    & 10   &    0.997 &       7,860 &               00:01 &                       00:02:22 \\
    & 100  &    0.996 &      53,988 &               00:11 &                       00:18:07 \\
    & 500  &    0.997 &     268,552 &               01:04 &                       01:31:18 \\
    & 1000 &    0.998 &     535,624 &               04:01 &                       02:56:21 \\
100 & 5    &      1.0 &       4,944 &               00:00 &                       00:01:13 \\
    & 10   &    0.999 &      12,444 &               00:01 &                       00:04:05 \\
    & 100  &    0.999 &     128,824 &               00:25 &                       00:41:54 \\
    & 500  &    0.999 &     676,768 &               04:21 &                       04:15:49 \\
    & 1000 &    0.999 &    1,362,220 &               12:10 &                       08:08:58 \\
\bottomrule
\end{tabular}
    \end{subtable}
    \begin{subtable}{.8\linewidth}
      \centering
        \caption{LGT}
\begin{tabular}{lrcrcc}
\toprule
 $\max L$   &  $M$    & Accuracy &  Num. data & Training (min) & Data gen. (hour/core) \\
\midrule
20  & 5    &    0.974 &        768 &               00:01 &                       00:00:19 \\
    & 10   &    0.994 &       1,548 &               00:02 &                       00:00:41 \\
    & 100  &    0.976 &      12,244 &               00:09 &                       00:04:20 \\
    & 500  &    0.975 &      58,900 &               00:24 &                       00:19:13 \\
    & 1000 &    0.975 &     118,104 &               00:27 &                       00:35:38 \\
50  & 5    &    0.997 &       2,952 &               00:01 &                       00:00:43 \\
    & 10   &    0.995 &       3,796 &               00:03 &                       00:01:01 \\
    & 100  &    0.995 &      44,116 &               00:23 &                       00:14:01 \\
    & 500  &    0.994 &     219,472 &               01:39 &                       01:06:45 \\
    & 1000 &    0.994 &     421,204 &               02:45 &                       02:10:45 \\
100 & 5    &    0.996 &       5,080 &               00:06 &                       00:01:23 \\
    & 10   &    0.996 &       7,540 &               00:05 &                       00:01:58 \\
    & 100  &    0.998 &     114,900 &               00:31 &                       00:34:25 \\
    & 500  &    0.998 &     605,652 &               04:44 &                       02:54:15 \\
    & 1000 &    0.998 &    1,175,628 &               10:23 &                       05:31:13 \\
\bottomrule
\end{tabular}
    \end{subtable} 
\end{table}

\begin{sidewaysfigure}
\centering
       \subfloat[Normal ML. Test accuracy = 0.815]
       {
        \includegraphics[height=7cm]{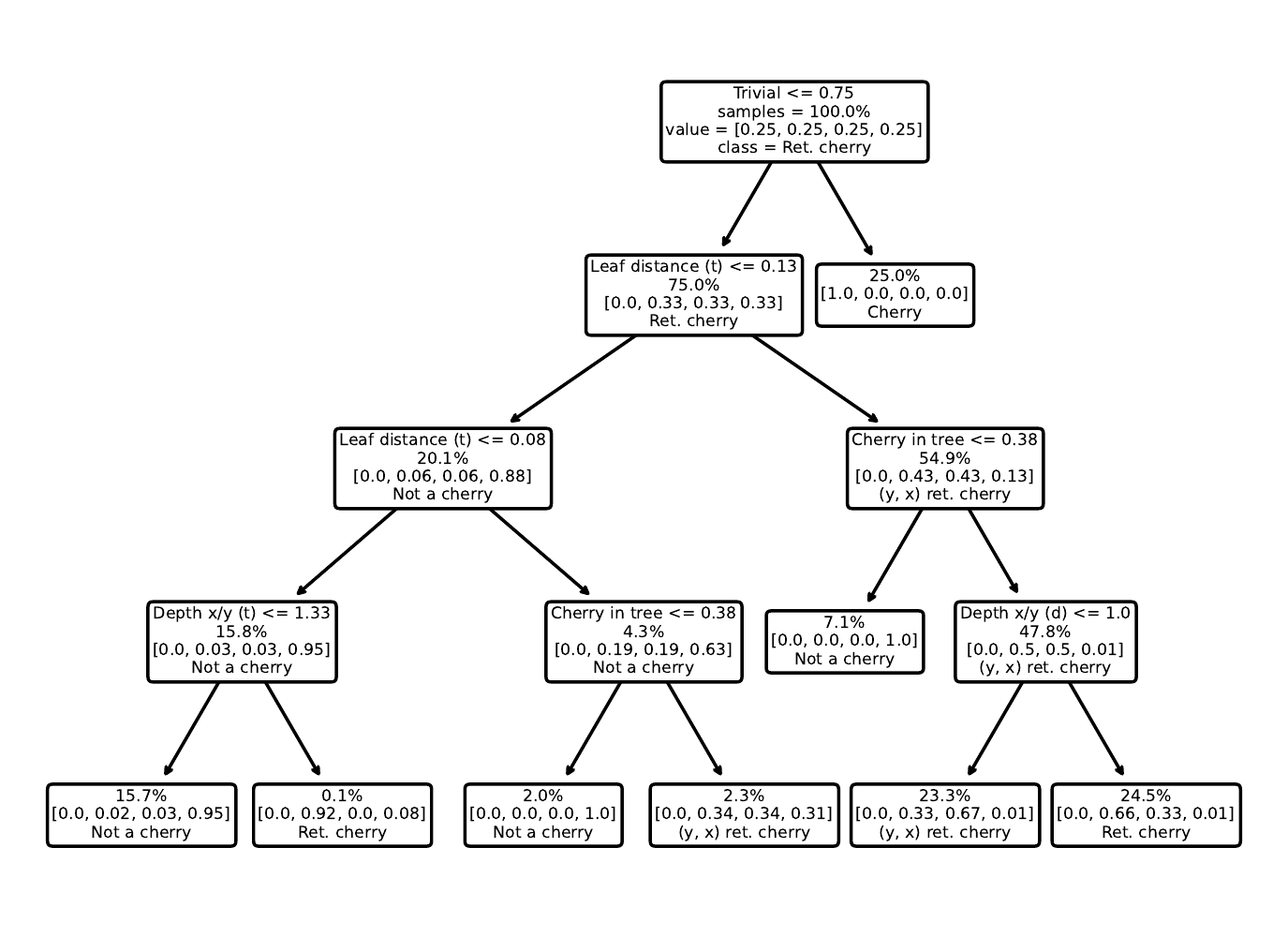}
       } 
       \\
        \subfloat[LGT ML. Test accuracy = 0.802]
        {
         \includegraphics[height=7cm]{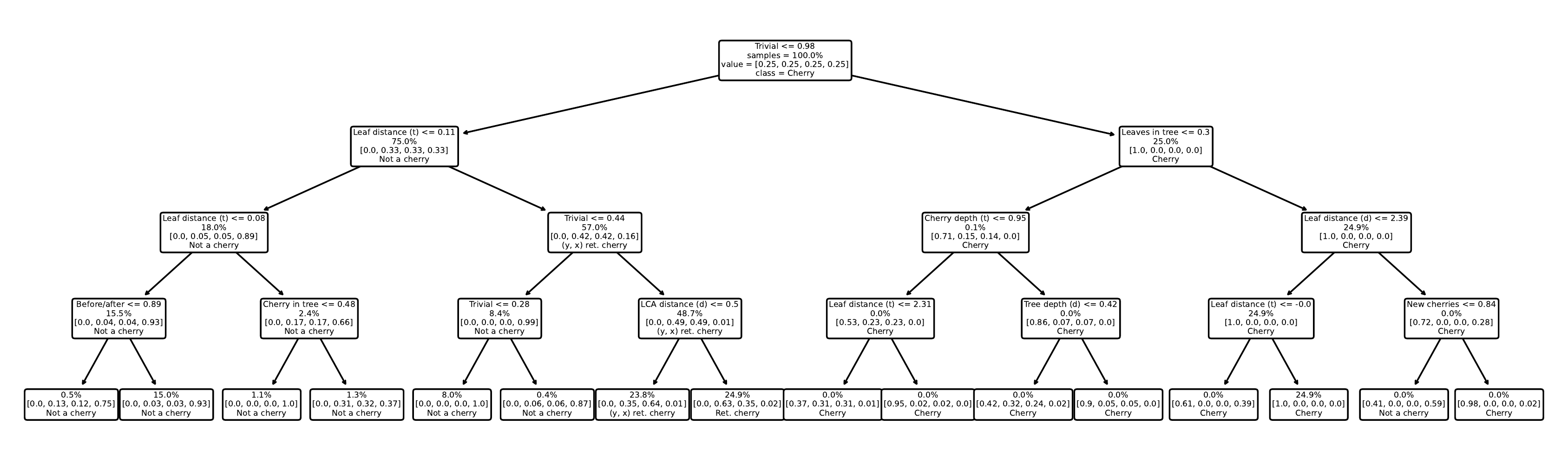}
        }
    \caption{Classification tree with depth 4 of (a) the normal data set and (b) the LGT data set. For each node in the trees, except for the terminal ones, the first line is the feature condition. If this condition is met by a data point, it traverses to the left child node, otherwise to the right one. In the terminal nodes this line is omitted as there is no condition given. In each node, as also indicated with labels in the root node, the second line `samples' is the proportional number of samples that follow the YES/NO conditions from the root to the parent of that node during the training process. The `value' list gives the proportion of data points in each class, compared to the sample of that node. The last line indicates the most dominant class of that node. If a data point reaches a terminal node, the observation will be classified as the indicated class.}
\label{fig:classificationTree}
 \end{sidewaysfigure}
\newpage
\section{Heuristic Performance of ML Models}\label{app:Ml_models_heatmaps}
\begin{figure}[h]
    \centering
        \subfloat[Normal ML]
        {
        \includegraphics[width=.5\columnwidth]{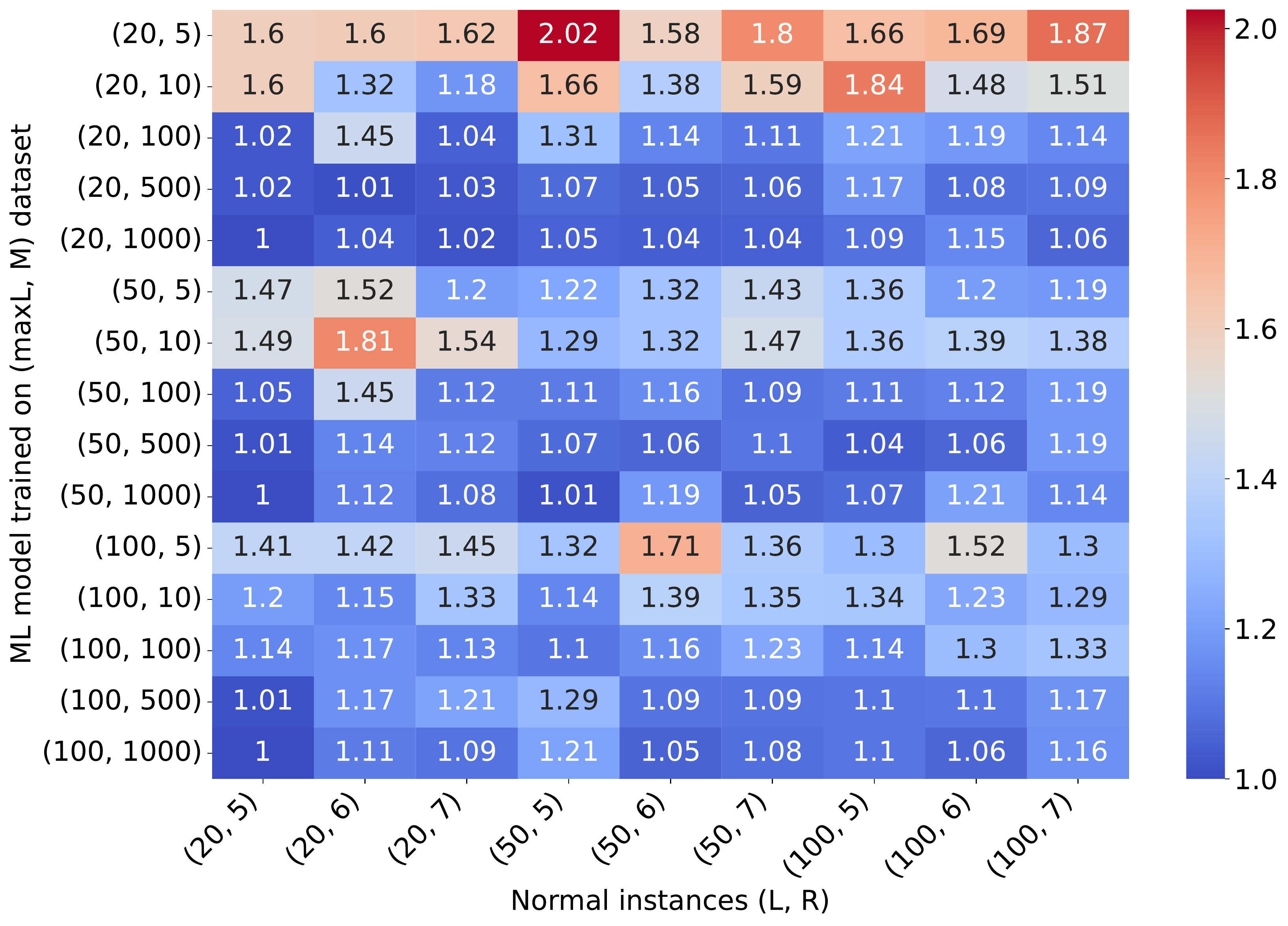}
        }
        \subfloat[LGT ML]
        {
        \includegraphics[width=.5\columnwidth]{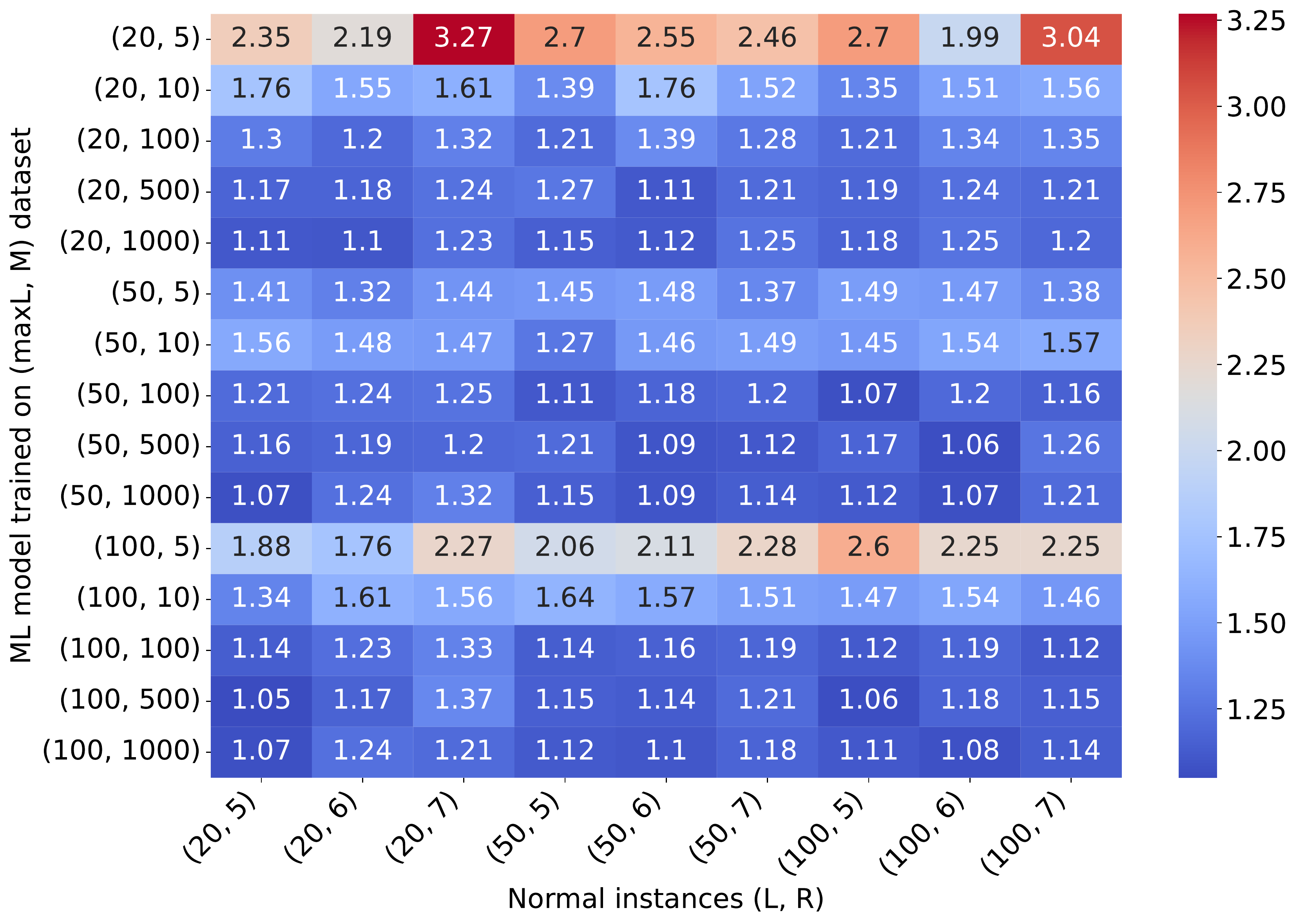}
        }
    \caption{Results for \ML{} on normal instances with the random forest model trained on each of the datasets given in Table~\ref{tab:ml_description}, where \textbf{(a)} gives the results when the ML model is trained on normal data, and \textbf{(b)} gives the results when the model is trained on LGT data. For each training dataset, identified by the parameter pair $(\max L, M)$, the value shown in the heatmap is the average, within each instance group, of the reticulation number found by \ML{} divided by the reference value. We used a group of 16 instances for each combination of parameters $L \in \{20, 50, 100\}$ and $R \in \{5, 6, 7\}$.}\label{fig:ml_heur_perf_norm} 
\end{figure}

\begin{figure}[h]
    \centering
      \subfloat[Normal ML]
      {
       \includegraphics[width=\columnwidth]{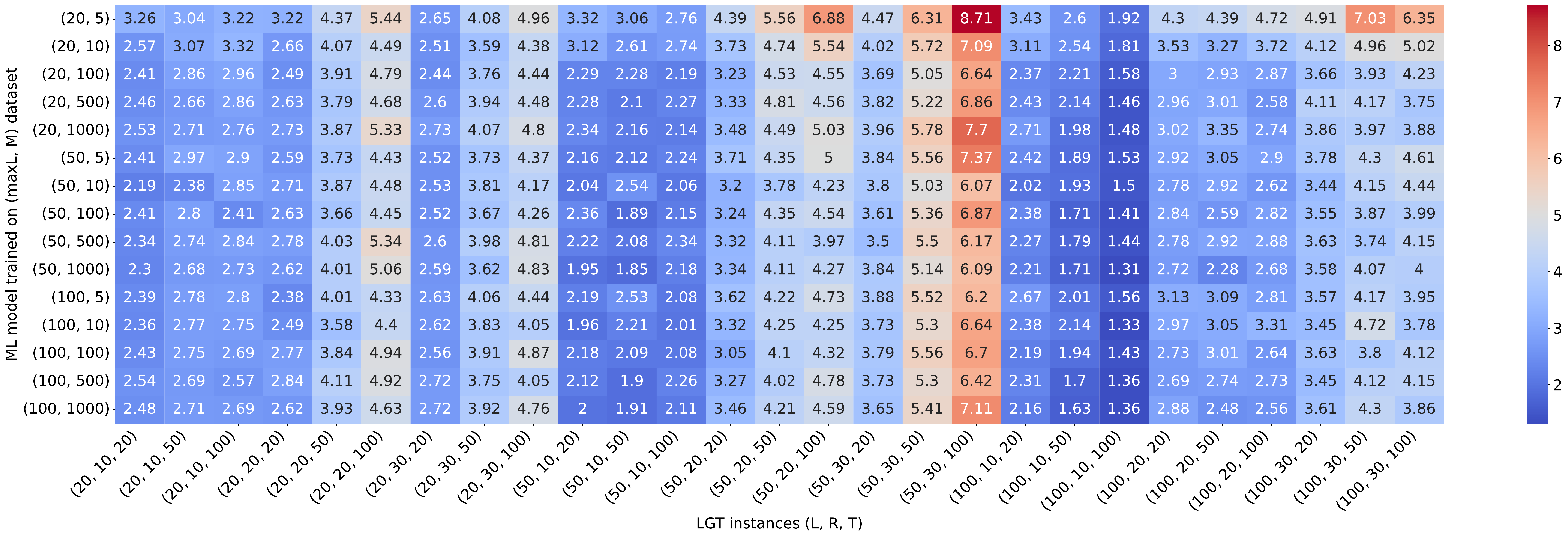}
      } \\
      \vspace{.5cm}
       \subfloat[LGT ML]
       {
        \includegraphics[width=\columnwidth]{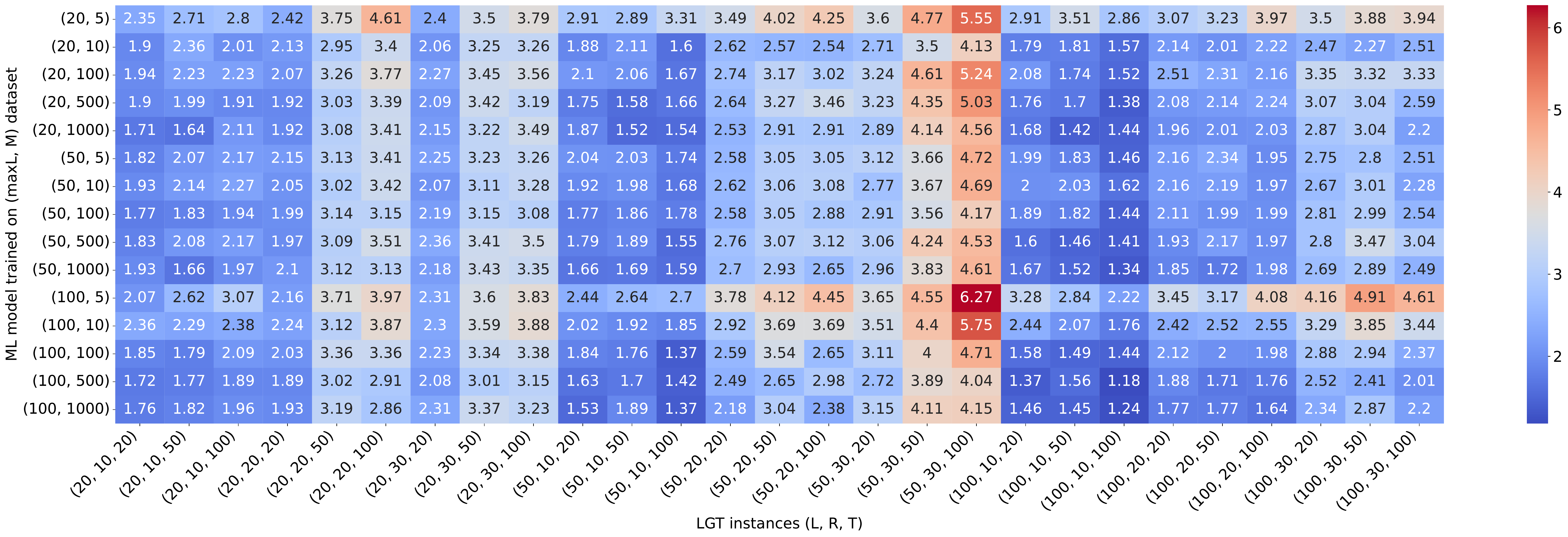}
       }
    \caption{Results for \ML{} on LGT instances for different training datasets, similar as description of Fig. \ref{fig:ml_heur_perf_norm}, with $L \in \{20, 50, 100\}$, $R \in \{10, 20, 30\}$ and $|\cT| \in \{20, 50, 100\}$.}\label{fig:ml_heur_perf_LGT}
\end{figure}

\begin{figure}[h]
    \centering
        \subfloat[Normal ML]
        {
        \includegraphics[width=\columnwidth]{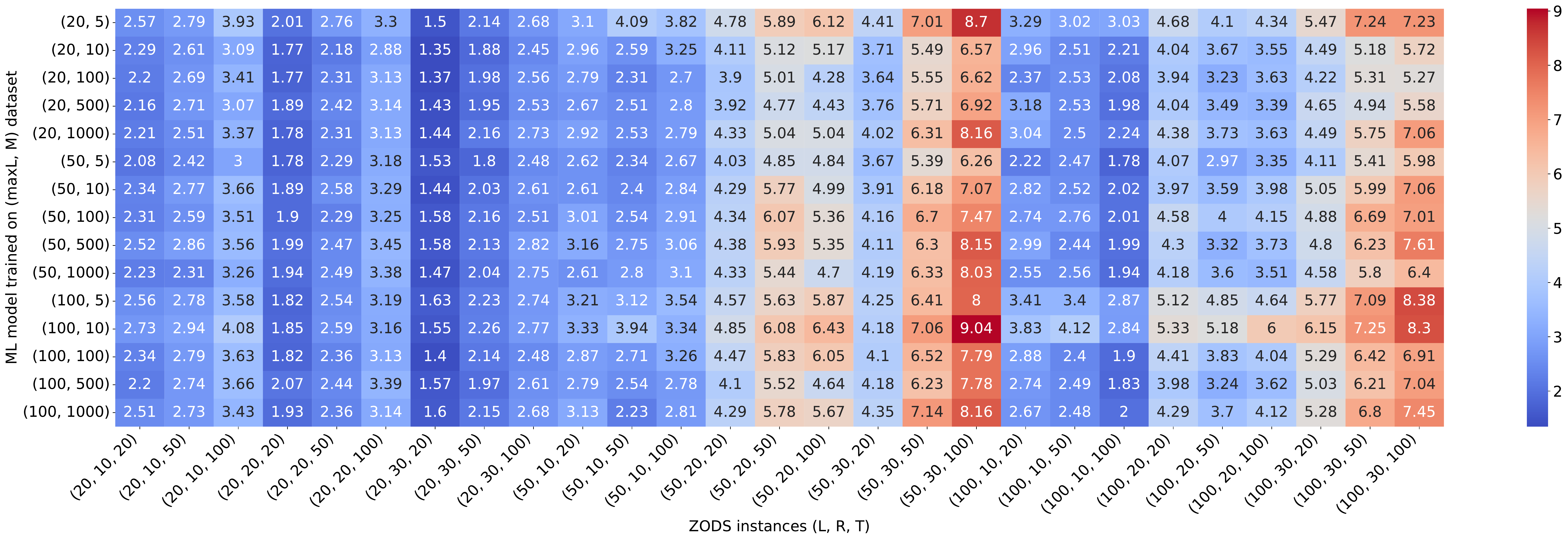}
        } \\
       \vspace{.5cm}
       \subfloat[LGT ML]
       {
       \includegraphics[width=\columnwidth]{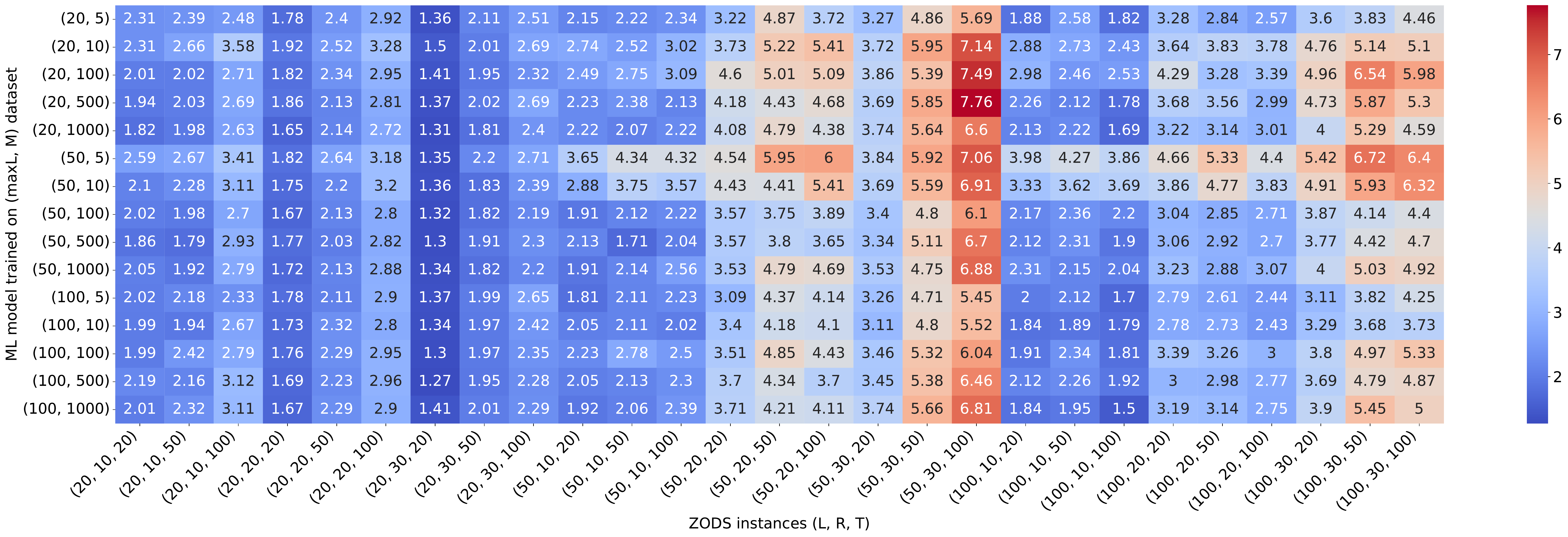}
       }
    \caption{Results for \ML{} on ZODS instances for different training datasets, similar as description of Fig. \ref{fig:ml_heur_perf_norm}, with $L \in \{20, 50, 100\}$, $R \in \{10, 20, 30\}$ and $|\cT| \in \{20, 50, 100\}$.}\label{fig:ml_heur_perf_ZODS}
\end{figure}
\clearpage

\bibliographystyle{plain}
\bibliography{bibliography.bib}
\end{document}